\documentclass[a4paper,UKenglish]{llncs}
\usepackage{etex}
\usepackage{microtype}
\pdfoutput=1

\usepackage{amsmath}
\usepackage{amssymb}
\usepackage{mathtools}

\usepackage{tikz}

\usetikzlibrary{decorations.pathmorphing}
\usetikzlibrary{decorations.pathreplacing}
\usetikzlibrary{shapes.arrows}

\usetikzlibrary{calc}
\tikzset{snake it/.style={decorate, decoration=snake}}
\tikzset{small snake/.style={decorate, decoration={snake, segment length=2mm, amplitude=1mm}}}

\let\originalleft\left
\let\originalright\right
\renewcommand{\left}{\mathopen{}\mathclose\bgroup\originalleft}
\renewcommand{\right}{\aftergroup\egroup\originalright}

\usepackage{framed}

\newcommand{\eps}{\varepsilon}

\newcommand{\cnot}{\mathrm{CNOT}}
\newcommand{\X}{\mathrm{X}}
\newcommand{\Y}{\mathrm{Y}}
\newcommand{\Z}{\mathrm{Z}}
\renewcommand{\H}{\mathrm{H}}
\renewcommand{\P}{\mathrm{P}}
\newcommand{\T}{\mathrm{T}}
\newcommand{\I}{\mathrm{I}}

\newcommand{\C}{\mathcal{C}}

\newcommand{\INQC}{\mathrm{INQC}} 
\newcommand{\gh}{\mathit{GH}}                
\newcommand{\GIP}{G_{\mathrm{IP}}}
\newcommand{\etaerr}{\eta_{\mathrm{err}}}
\newcommand{\etaloss}{\eta_{\mathrm{loss}}}

\DeclarePairedDelimiter{\norm}{\lVert}{\rVert}

\usepackage{qcircuit}

\newcommand{\ket}[1]{{\left\vert{#1}\right\rangle}}

\usepackage{hyperref}
\hypersetup{
    colorlinks=true,
    linkcolor=black,
    citecolor=black,
    filecolor=black,
    urlcolor=black,
}

\title{Instantaneous non-local computation of low T-depth quantum circuits}
\author{Florian Speelman}
\institute{CWI, Amsterdam}

\begin{document}

\maketitle

\begin{abstract}
Instantaneous non-local quantum computation requires multiple parties to jointly perform a quantum operation,
using pre-shared entanglement and a single round of simultaneous communication. We study this task
for its close connection to position-based quantum cryptography, but it also has natural applications in the context of foundations of quantum physics and in distributed computing.
The best known general construction for instantaneous non-local quantum computation requires a pre-shared state which is exponentially large in the number of qubits involved in the operation,
while efficient constructions are known for very specific cases only.

We partially close this gap by presenting new schemes for efficient instantaneous non-local computation of several classes of quantum circuits, using the Clifford+T gate set.
Our main result is a protocol which uses entanglement exponential in the T-depth of a quantum circuit, able to perform non-local computation
of quantum circuits with a (poly-)logarithmic number of layers of T gates with quasi-polynomial entanglement.
Our proofs combine ideas from blind and delegated quantum computation with the garden-hose model, a combinatorial model of communication complexity which was recently introduced as a tool
for studying certain schemes for quantum position verification.
As an application of our results,
we also present an efficient attack on a recently-proposed scheme for position verification by Chakraborty and Leverrier.
\end{abstract}

\section{Introduction}
We study the
task of instantaneous non-local quantum computation, and
present new protocols to efficiently perform this task for specific classes of quantum circuits. Our main motivation
comes from position-based quantum cryptography, where previous attacks on schemes for position-based quantum
cryptography have taken either of two forms:

First results on quantum position-based cryptography
involved attacks on specific proposals for schemes, such
as the attacks by Lau and Lo~\cite{LL11}, those by Kent, Munro and Spiller~\cite{KMS11}, and the attack on Beigi and K\"onig's scheme using mutually-unbiased-bases~\cite{Speelman11}.  A certain family of efficient attacks on a concrete class of single-qubit schemes~\cite{BFSS13}
was formalized by the garden-hose model. 
Described as `fast protocols for bipartite unitary operators',
Yu, Griffiths and Cohen~\cite{YGC12,Yu11} give protocols that, although not directly inspired by position-based quantum cryptography,
can be translated to our setting.

On the other hand
Buhrman et al.~\cite{Buhrman2011} constructed a general attack which treats
the quantum functionality of the protocol to be attacked as a black box. For a protocol which uses a message of $n$ qubits, 
the entanglement consumption of
this attack is around $2^{ \log{(\frac{1}{\eps})} 2^{4n} }$ EPR pairs\index{EPR pair}, doubly exponential in $n$. Here $\eps$ represents the probability 
that the attack does not succeed. The construction
of Buhrman et al.~was based on a protocol for `instantaneous non-local measurement' by Vaidman~\cite{Vaidman03,CCJP10}.
Beigi and K\"onig~\cite{Beigi2011} later
constructed a more efficient general attack, using port-based teleportation\index{teleportation!port-based} -- a new teleportation method introduced by Ishizaka and Hiroshima~\cite{IH08,IH09}.
The improved attack uses $O(n \frac{2^{8 n}}{\eps^2})$ EPR pairs, still an exponential dependence on $n$. 

These protocols were able to solve the following 
task. 
Given a constant $\eps \geq 0$ and an $n$-qubit quantum operation\footnote{Our constructions only consider unitaries given by quantum circuits, but the task naturally extends to more general quantum operations.
The motivation for Vaidman's original scheme~\cite{Vaidman03}, which formed the basis of Buhrman et al.'s construction, was to instantaneously perform a non-local measurement. Our
constructions can also be applied to that case, by writing the measurement as a unitary operation followed by a measurement in the computational basis.} $U$, where $n$ is a natural number.
Two players, Alice and Bob, receive an arbitrary input state $\rho_{AB}$ of $n$ qubits, with the players receiving $n/2$ qubits each.
After a single round of simultaneous quantum\footnote{Since restriction to classical communication is not necessarily dictated by the application in
position-based quantum cryptography, we allow quantum communication. All presented protocols work equally well when all messages are classical instead.} communication, the players must output
a state $\eps$-close to $U \rho_{AB} U^{\dagger}$. Alice
outputs the first $n/2$ qubits of the state and Bob outputs the other $n/2$ qubits.
We define $\INQC_{\eps}(U)$ as the smallest number of EPR pairs that the players have to share
at the start of a protocol which performs this task.
$\INQC(U)$ is used as a shorthand for $\INQC_0(U)$, a protocol which works with no error.
We present a more precise definition of $\INQC$ is presented in Appendix~\ref{sec:definqc}.

In this work we partially bridge the gap between
efficient specific constructions for instantaneous non-local computation
and expensive general ones, by constructing a protocol for non-local computation of a unitary transformation $U$
such that the entanglement use of the protocol depends on the quantum circuit which describes $U$.

In particular, writing quantum circuits over the Clifford+T gate set, we create a protocol using entanglement exponential in the \emph{T-count}.
We also present a protocol that uses an amount of entanglement which scales as the number of qubits $n$ raised to the power of the \emph{T-depth} of the circuit. 
Even though
this is a quickly-growing dependence, for circuits of constant T-depth this
amounts to a polynomial dependence on $n$, unlike any earlier construction. For circuits of polylogarithmic T-depth we obtain an amount of entanglement which is quasi-polynomial in $n$, i.e.\ a dependence of the form $2^{(\log n)^c}$ for some constant $c$.
Note that the depth and size of the quantum circuit can be much higher than its T-depth: we allow an arbitrary number of gates
from the Clifford group in addition to the limited number of T gates.
Our results imply new efficient attacks on any scheme for position-verification where the action of the honest party can be written as a low T-depth quantum circuit.

Linking blind quantum computation and instantaneous non-local quantum computation was first considered by Broadbent\footnote{These results were first available as privately-circulated notes in December 2011, and were made
available online in December 2015.}~\cite{Bro15a}, who considered a setting where the parties have access to non-local boxes -- correlations even stronger than those
allowed by quantum mechanics. 
The techniques we use are also based on delegated and blind quantum computation~\cite{Childs2005,AS06,DNS10,FBS+14,Bro15}
and results on computation via teleportation~\cite{GC99}, but we combine them with new ideas from the \emph{garden-hose model}~\cite{BFSS13,KP14}\index{garden-hose model} -- a recently-introduced combinatorial model for communication complexity with close links to a specific class of schemes for position verification.

We prove two main theorems, each improving on the entanglement consumption of the best-known previous constructions for non-local instantaneous quantum computation for specific circuits\footnote{From now on, whenever we write `quantum circuit', we will always
mean a quantum circuit that only uses the Clifford group generators,
together with T gates.}.
Additionally, we use our proof method to construct a new attack on a scheme for position verification which was recently proposed by Chakraborty and Leverrier~\cite{CL15}.

\textbf{Theorem~\ref{thm:tcount}.}
Any $n$-qubit Clifford+T quantum circuit $C$ which has at most $k$ $\T$ gates has a protocol for instantaneous non-local computation 
using $O(n 2^k)$ EPR pairs.

\textbf{Theorem~\ref{thm:tdepth}.}
Any $n$-qubit quantum circuit $C$ using the Clifford+T gate set which has $\T$-depth $d$, has a protocol for instantaneous non-local computation
using $O(\, (68n)^{d}\, )$ EPR pairs.
\\
\\
The main technical tool we use in the proof of our depth-dependent construction is the following lemma, which is able to remove a conditionally-applied gate from the Clifford group without any communication
-- at an entanglement cost which scales with the garden-hose complexity of the function which describes the condition.

\textbf{Lemma~\ref{lem:nonlocalphase}.}
Let $f:\{0,1\}^n \times \{0,1\}^n \to \{0,1\}$ be a function known to all parties,
and let $\gh(f)$ be the garden-hose complexity of the function $f$. 
Assume Alice has a single qubit with state $\P^{f(x,y)} \ket{\psi}$, for
binary strings $x,y\in \{0,1\}^n$, where Alice knows the string $x$ and Bob knows $y$.
The following two statements hold: 
\begin{enumerate}
\item There exists an instantaneous protocol without any communication which uses
$2 \gh(f)$ pre-shared EPR pairs after which a
chosen qubit of Alice is in the state
$\X^{g(\hat{x},\hat{y})} \Y^{h(\hat{x},\hat{y})} \ket{\psi}$. Here
$\hat{x}$ depends only on $x$ and the $2\gh(f)$ bits that describe the
measurement outcomes of Alice, and $\hat{y}$ depends on $y$ and
the measurement outcomes of Bob.
\item The garden-hose complexities of the functions $g$ and $h$ are at most linear in the garden-hose complexity of the function $f$. More precisely, $\gh(g) \leq 4 \gh(f) + 1$ and $\gh(h) \leq 11 \gh(f) + 2$.
\end{enumerate}

\noindent \\ Chakraborty and Leverrier~\cite{CL15} recently proposed a protocol for quantum position verification on the interleaved multiplication of unitaries.
They show that all known attacks, applied to this protocol, require entanglement exponential in the number of terms $t$ in the product.
As an application of Lemma~\ref{lem:nonlocalphase}, we present an attack on their proposed protocol which has entanglement cost polynomial in $t$ and the number of qubits $n$.
The new attack requires an amount of entanglement which scales as $(\frac{t}{\eps})^{O(1)}$ per qubit, and for each qubit succeeds with probability at least $1-\eps$.

\section{Preliminaries}
\subsection{The Pauli matrices and the Clifford group}\index{Pauli matrix}
\index{Clifford group}
The single-qubit \emph{Pauli matrices} are $\X=\begin{pmatrix}0 & 1 \\1 & 0\end{pmatrix}$, $\Y=\begin{pmatrix}0 & -i \\i & 0\end{pmatrix}$,
$\Z=\begin{pmatrix}1 & 0 \\0 & -1\end{pmatrix}$, and the identity
$\mathrm{I}=\begin{pmatrix}1 & 0 \\0 & 1\end{pmatrix}$.
A \emph{Pauli operator} on an $n$-qubit state is the tensor product of $n$ one-qubit Pauli matrices, the group of $n$ qubit Pauli operators\footnote{The given definition includes a global phase, which is not important when viewing the elements as quantum gates.} is
$\mathcal{P} = \{\sigma_1 \otimes \dots \otimes \sigma_n \mid \forall j: \sigma_j \in \{I,X,Y,Z\} \} \times \{\pm 1, \pm i\}$. These are some of the simplest quantum operations and appear, for example, as corrections for standard quantum teleportation.

The \emph{Clifford group} can be defined as those operations that take elements of the Pauli
group to other elements of the Pauli group under conjugation -- the \emph{normalizer} of the Pauli group. We consider
the Clifford group on $n$ qubits, for some natural number $n$.
\begin{equation}\label{eq:defclifford}
\mathcal{C} = \{U \in \mathcal{U}(2^n) \mid \forall \sigma : \sigma \in \mathcal{P} \implies U \sigma U^\dagger \in \mathcal{P} \}
\end{equation}
Notable elements of the Clifford group are the single-qubit gates given by the Hadamard matrix 
$\H=\frac{1}{\sqrt{2}}\begin{pmatrix}1 & 1 \\1 & -1\end{pmatrix}$
and the phase gate 
$
\P=\begin{pmatrix}1 & 0 \\0 & i\end{pmatrix} $,
 and
the two-qubit CNOT gate given by 
$\cnot=\begin{pmatrix}1 & 0 & 0 & 0 \\0 & 1 & 0 & 0 \\ 0 & 0 & 0 & 1 \\ 0 & 0 & 1 & 0\end{pmatrix}$.

The set $\{\H,\P,\cnot\}$
generates the Clifford group up to a global phase when applied to 
arbitrary qubits, see e.g.~\cite{Got98}.
For all these gates, we will use subscripts to indicate the qubits or wires to which they are applied; e.g.~$\H_j$ is a Hadamard gate applied to the $j$-th
wire, and $\cnot_{j,k}$ is a $\cnot$ that has wire $j$ as control and $k$ as target.

Even though there exist interesting quantum circuits that use only gates
from the Clifford group, it is not a universal set of gates. Indeed, the Gottesman--Knill states that such a circuit can be efficiently simulated by a classical computer, something which is not known to be true for general quantum circuits~\cite{Gottesman98,AG04}. By extending $\mathcal{C}$ with \emph{any} gate, we do obtain a gate-set which is universal for quantum computation~\cite{NRS01}. 

The gate we will use to extend the Clifford gates to a universal set
is the T gate, sometimes called $\pi/8$-gate or $\mathrm{R}$, defined by $\T=\begin{pmatrix}1 & 0 \\0 & e^{i \pi / 4}\end{pmatrix}$.
We will write all circuits using gates from the set
$\{\X, \Z, \H, \P, \cnot, \T\}$.
Technically $\X$, $\P$, and $\Z$ are redundant here,
since they can be formed by the others as $\P=\T^2$, $\Z=\P^2$ and $\X=\Z\H\Z$, but 
we include them for convenience.

In our protocols for
instantaneous non-local computation, we will alternate teleportation steps with gate operations, and therefore the interaction between the Pauli matrices and the other gates are especially important. We will make much use of the following identities,
which can all be easily checked\footnote{Here equality is up to a global phase -- which we will ignore from now on for simplicity.}.
\begin{equation}\label{eq:pauliclifford}
\begin{aligned}[c]
\X\Z &= \Z\X \\
\P\Z &= \Z\P \\
\P\X &= \X\Z\P \\
\end{aligned}
\quad\quad
\begin{aligned}[c]
\H\X &= \Z\H \\
\H\Z &= \X\H\\
\T\X &= \P\X\T \\
\end{aligned}
\quad\qquad
\begin{aligned}[c]
\cnot_{1,2} (\X\otimes\I) &= (\X\otimes\X) \cnot_{1,2}  \\
\cnot_{1,2} (\I\otimes\X) &= (\I\otimes\X) \cnot_{1,2}  \\
\cnot_{1,2} (\Z\otimes\I) &= (\Z\otimes\I) \cnot_{1,2}  \\
\cnot_{1,2} (\I\otimes\Z) &= (\Z\otimes\Z) \cnot_{1,2}  \\
\end{aligned}
\end{equation}

\subsection{Key transformations from Clifford circuits}\label{sec:paulitransform}
For a $0/1$ vector $v$ of length $n$
and for any single-qubit operation $U$, we write $U^{v} = \bigotimes^{n}_{j=1} U^{v_j}$, i.e., $U^{v}$ is the application of $U$ on all qubits $j \in [n]$ for which $v_j=1$.
When Alice teleports a state $\ket{\psi}$ of $n$ qubits to Bob, the uncorrected state at Bob's side can be written as $\X^{a_x} \Z^{a_z} \ket{\psi}$. Here
we let $a_x$ and $a_z$ be the vectors representing the outcomes of the Bell measurements\index{Bell measurement} of Alice.
In analogy with the the literature on assisted and blind quantum computation, we
will call the teleportation measurement outcomes $a_x$ and $a_z$ the \emph{key} needed to decode $\ket{\psi}$.

The specific entries of these keys will often depend on several different measurement outcomes, given by earlier steps in the
protocol, and we will therefore occasionally describe them as \emph{polynomials} over $\mathbb{F}_2$.
Viewing the keys as polynomials is especially helpful
in the description of the more-complicated protocol of Section~\ref{sec:tdepth}.

For any gate from the Clifford group $U \in \mathcal{C}$, if we apply $U$ on the
encoded state, we can describe the resulting state as $U\ket{\psi}$ with a new key.
That is, $U X^{a_x} Z^{a_z} \ket{\psi} = X^{\hat{a}_x} Z^{\hat{a}_z} U \ket{\psi}$ for
some new 0/1 keys $\hat{a}_x, \hat{a}_z$.
The transformations of the keys will have a particularly simple form. (See for example~\cite{BCL+06} for a characterization of these transformations and a different application of Clifford circuit computation.)

For example, we can write the identities of Equation~\ref{eq:pauliclifford}
in terms of key transformations. The transformations that occur when a bigger Pauli operator is
applied, can then be easily found by writing the Pauli operator in terms of its generators $\{\H,\P,\cnot\}$,
and applying these rules one-by-one.
We will write $(x_1, x_2 \mid z_1, z_2)$ as a shorthand for, respectively, the X key
on the first and second qubit, and the Z key on the first and second qubit -- this is a convenient notation\footnote{This
mapping is called the symplectic notation when used in the stabilizer formalism, although we won't need to introduce the associated symplectic inner product for our construction.}
for the pair of vectors $a_x$ and $a_z$ that represent these keys.
All addition of these keys will be 
over $\mathbb{F}_2$, i.e., the $+$ represents the binary exclusive or.
\begin{align*}
\P (x \mid z) \to&\, (x \mid x + z) \P  \\
\H (x \mid z) \to&\, (z \mid x) \H  \\
\cnot_{1,2}  (x_1, x_2 \mid z_1, z_2) \to&\, (x_1, x_1 + x_2 \mid z_1 + z_2, z_2) \cnot_{1,2} 
\end{align*}

\subsection{Clifford+T quantum circuits, T-count and T-depth}
In several different areas of quantum information, gates
from the Clifford group are `well-behaved' or `easy', while the other non-Clifford
gates are hard -- an observation which was also made, with several examples, in the recent~\cite{BJ15}. 

The \emph{T-count}\index{T-count|textbf} of a quantum circuit is defined as the number of $\T$ gates in the entire
quantum circuit. The \emph{T-depth}\index{T-depth|textbf} is the number of layers of $\T$ gates, when viewing the circuit
as alternating between Clifford gates and a layer of simultaneous T gates.
See for example Figure~\ref{fig:tdepth}.

Given a quantum operation, it is not always obvious what is the best circuit in terms of T-count or T-depth. Recent work gave algorithms for finding circuits that are optimized in terms of T-depth~\cite{AMMR13,GS13,Sel13,AMM14} and optimal constructions for
arbitrary single-qubit unitaries have also been found~\cite{KMM13,RS14,Sel15}. These constructions 
sometimes increase the number of qubits involved by adding ancillas---the use of which can
greatly decrease the T-depth of the resulting circuit.

\subsection{The garden-hose model}\label{sec:nonlocalip}\index{garden-hose model}
The garden-hose model is a combinatorial model of communication complexity,
first introduced by Buhrman, Fehr, Schaffner and Speelman~\cite{BFSS13}. The recent work by Klauck and Podder~\cite{KP14} further investigated the notion, proving several follow-up results.
Here we repeat the basic definitions of the garden-hose model and its link to attacks on schemes for position-based quantum cryptography.

Alice has an input $x \in \{0,1\}^n$, Bob has an input $y \in \{0,1\}^n$, and the players
want to compute a function $f : \{0,1\}^n \times \{0,1\}^n \to \{0,1\}$ in the following way. Between the two players are $s$ pipes, 
and, in a manner depending on their respective inputs, the players link up these pipes one-to-one with hoses. Alice also has 
a water tap, which she can connect to one of these pipes. When $f(x,y)=0$, the water should exit on Alice's side, and when $f(x,y)=1$
we want the water to exit at Bob's side. The garden-hose complexity of a function $f$, written $\gh(f)$, then is the 
least number $s$ of pre-shared pipes the players need to compute the function in this manner.

There is a natural translation from strategies of the garden-hose game to a quantum protocol that routes a qubit to either Alice
or Bob depending on their local inputs, up to teleportation corrections. 
Consider the following quantum task, again dependent on a function $f$ like in the previous paragraph.
Alice now receives a quantum state $\ket{\psi}$ and a classical input $x$, Bob receives input $y$, and the players
are allowed one round of simultaneous communication. If $f(x,y)=0$, Alice must output $\ket{\psi}$ after this round of communication,
and otherwise Bob must output $\ket{\psi}$. We would like to analyze how much pre-shared entanglement the players need to perform this task.

From the garden-hose protocol for $f$, the players can come up with a strategy for this quantum task that needs at most $\gh(f)$ EPR pairs pre-shared.
Every pipe corresponds to an EPR pair. If a player's garden-hose strategy dictates a hose between some pipe $j$ and another pipe $k$,
then that player performs a Bell measurement of EPR-halves labeled $j$ and $k$. Alice's connection of the water tap to a pipe corresponds
to a Bell measurement between her input state $\ket{\psi}$ and the local half of an EPR pair.
After their measurements, the correct player will hold the state $\ket{\psi}$, up to Pauli corrections incurred by the teleportations.
The corrections can be performed after a step of simultaneous communication containing the outcomes of all measurements.

We will describe some of the logic in terms of the garden-hose model, as an 
abstraction away from the qubits involved.
When we refer to a quantum implementation of a garden-hose strategy, we always mean
the back-and-forth teleportation as described above.

The following lemma will prove to be useful. Let the number of \emph{spilling pipes} of a garden-hose protocol for a player
be the number of possible places the water could possibly exit. That is, the number of spilling pipes for Alice for a specific $x$, is the number of different places the water could exit on her side
over all Bob's inputs $y$. The number of spilling pipes for Alice is then the maximum number of spilling pipes over all $x$.
To be able to chain different parts of a garden-hose protocol together, it can be very convenient to only have a single spilling pipe for each player.

\begin{lemma}[Lemma 11 of~\cite{KP14}]\label{lem:ghsingleoutput}
A garden-hose protocol $P$ for any function $f$ with multiple spilling pipes can be converted to another garden-hose protocol $P'$ for $f$ that has only
one spilling pipe on Alice's side and one spilling pipe on Bob's side. The size of $P'$ is at most 3 times the size of $P$ plus 1.
\end{lemma}

Klauck and Podder also showed that computing the binary XOR of several protocols is possible with only a linear overhead in total garden-hose complexity~\cite[Theorem 18]{KP14}. 
We give an explicit construction for this statement in Appendix\ref{sec:proofsumgh} --
the result already follows from the similar construction of~\cite[Lemma 12]{KP14}, except that we 
obtain a constant which is slightly better than unfolding their (more general) proof.
\begin{lemma}\label{lem:sumgh}
Let $(f_1, f_2, \dots, f_k)$ be functions, where each function $f_i$ has garden-hose complexity $\gh(f_i)$.
Let $c \in \{0,1\}$ be an arbitrary bit. Then, 
\[
\gh \left( c \oplus \bigoplus _{i=1} ^{k} f_i \right) \leq 4 \sum_{i=1}^k \gh(f_i) + 1 \, .
\]
\end{lemma}

\section{Low T-count quantum circuits}\label{sec:tcount}
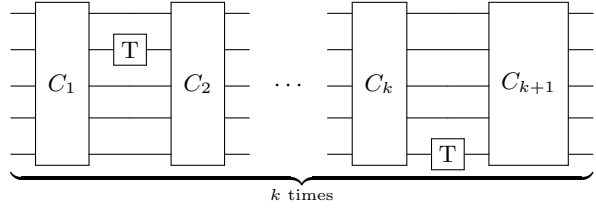
\begin{figure}[thb!]
\begin{center}
$
\underbrace{
\Qcircuit @C=1em @R=.4em {
& \multigate{4}{C_1} & \qw & \multigate{4}{C_2} & \qw & & & \multigate{4}{C_{k}} & \qw & \multigate{4}{C_{k+1}} & \qw \\
& \ghost{C_1} 		 & \gate{\T} & \ghost{C_2} 		  & \qw & & & \ghost{C_{k}} & \qw & \ghost{C_{k+1}} & \qw\\
& \ghost{C_1} 		 & \qw & \ghost{C_2} 		  & \qw & \push{\cdots} & & \ghost{C_{k}} & \qw & \ghost{C_{k+1}} & \qw\\
& \ghost{C_1} 		 & \qw & \ghost{C_2} 		  & \qw & & & \ghost{C_{k}} & \qw & \ghost{C_{k+1}} & \qw\\
& \ghost{C_1} 		 & \qw & \ghost{C_2} 		  & \qw & & & \ghost{C_{k}} & \gate{\T} & \ghost{C_{k+1}} & \qw\\
}
}_{\text{$k$ times}}
$
\end{center}
\caption{A circuit with T-count $k$. The $C_i$ gates represent subcircuits consisting only of operation from the Clifford group $\C$.}
\label{fig:tcount}
\end{figure}

\begin{theorem}\label{thm:tcount}
Let $C$ be an $n$-qubit quantum circuit with gates from the Clifford+T gate set, and let $C$ contain $k$
$\T$-gates in total. Then $\INQC(C) \leq O(n 2^k)$, i.e.,
there exists a protocol for two-party instantaneous non-local computation of $C$
which uses a pre-shared entangled state of $O(n 2^k)$ EPR pairs.
\end{theorem}
\begin{proof}
Let Alice's input state be some arbitrary quantum state $\ket{\psi_0}$.
We will write
the quantum state at step $t \in \{0, \dots, k\}$, as intermediate result of executing the circuit $C$ for $t$ steps, as $\ket{\psi_t}$. 
Let $C_t$ be the subcircuit, consisting only of Clifford gates, between the $(t-1)$th and $t$th T gates.
At step $t$, the circuit alternates between the Clifford subcircuit $C_t$ and a T-gate on some wire $w_t$ which we write as $T_{w_t}$, that is, we define $
T_{w_t} = \I^{\otimes w_t-1} \otimes \T \otimes \I^{\otimes n-w_t-1}$.

Because of the nature of the setting, all steps are done instantaneously unless
noted otherwise, without waiting for a message of the other party.
For example, if the description mentions that one party teleports a qubit, we
can instantly describe the qubit as `being on the other side', but the other party will act on the uncorrected qubit, since the communication will only happen afterwards and simultaneously.

We first give a high-level description of the protocol. Bob teleports his part of the state to
Alice, who holds the entire state -- up to teleportation corrections.
Alice will now apply the first set of Clifford gates, followed by a single
$\T$ gate.
The teleportation corrections (all known to Bob)
determine whether the $\T$ gate that Alice performs creates an unwanted extra $\P$ gate on the state. The extra $\P$ gate is created whenever
an X correction is present, because of the relation $\T \X = \P \X \T$. Therefore,
even though Alice holds the state, only Bob knows whether the state
has an extra unwanted $\P$ gate or not.

To remove the unwanted gate, Alice teleports all $n$ qubits back to Bob, who
corrects the phase gate (if present). The players then perform a garden-hose-like trick to keep the form of the key simple, at the cost of doubling the total size at each step.

Now we will give the precise description of the players' actions:

\begin{description}
\item[Step 0] Bob performs a Bell measurement to teleport
all his $n/2$ qubits to Alice, where we write the needed X-corrections as $b^0_{x,i}$ and Z-corrections $b^0_{z,i}$, for $i=n/2+1,\dots,n$. Now, since the qubits Alice already started with don't
need a correction, we have $b^0_{x,i}=b^0_{z,i}=0$ for $i=1, \dots, n/2$. Then we write $b^0_x$ and $b^0_z$ for the 0/1 vector containing the X corrections and Z correction respectively.
The complete state is $\X^{b^0_x}\Z^{b^0_z} \ket{\psi_0}$, where all qubits are at Alice's side while Bob knows the key.

\item[Step 1.a] Alice executes $C_1$ on the (uncorrected) qubits, so that the state now equals
\[
C_1 \X^{b^0_x}\Z^{b^0_z} \ket{\psi_0} = \X^{\hat{b}^1_x}\Z^{\hat{b}^1_z} C_1 \ket{\psi_0} \,,
\]
where $(\hat{b}^1_x, \hat{b}^1_z) = f_1(b^0_x, b^0_z)$, with $f_1: \mathbb{F}^{n}_2 \times \mathbb{F}^{n}_2 \to \mathbb{F}^{n}_2 \times \mathbb{F}^{n}_2$ is a formula that consists
of relabeling and addition over $\mathbb{F}_2$, and that is known to all parties. Bob knows all the entries of the vectors $\hat{b}^1_x$ and $\hat{b}^1_z$ that contain the new teleportation corrections.

\item[Step 1.b] Alice executes the $\T$ gate on the correct wire $w_1 \in \{1,\dots, n\}$ of
the uncorrected qubits. Define $\mathbf{b}^{1} = \hat{b}^1_{x,w_1}$, the $w_1$ entry of the vector $\hat{b}^1_x$. The state in Alice's possession is now
\begin{align*}
\T_{w_1} \X^{\hat{b}^1_x}\Z^{\hat{b}^1_{z}} C_1 \ket{\psi_0} 
&= \P^{\mathbf{b}^{1}}_{w_1} \X^{\hat{b}^1_x}\Z^{\hat{b}^1_{z}} \T_{w_1} C_1 \ket{\psi_0} 
= \P^{\mathbf{b}^{1}}_{w_1} \X^{\hat{b}^1_x}\Z^{\hat{b}^1_{z}} \ket{\psi_1}
\,.
\end{align*}
That is, besides the presence of the Pauli gates, depending on the teleportation measurements, the $w_1$ qubit possibly has an extra phase gate
that needs to be corrected before the protocol can continue.

\item[Step 1.c] Alice teleports all qubits to Bob, with teleportation outcomes
$a^1_{x}, a^1_{z} \in \mathbb{F}^{n}_2$. We will define the $\mathbf{a}^{1}$ as the $w_1$ entry of $a^{1}_{x}$. Bob then has the state
\[
\X^{a^1_{x}} \Z^{a^1_{z}} \P^{\mathbf{b}^{1}}_{w_1} \X^{\hat{b}^1_x}\Z^{\hat{b}^1_{z}} \ket{\psi_1}
= \P^{\mathbf{b}^{1}}_{w_1} \X^{\hat{b}^1_x}\Z^{\hat{b}^1_{z}} \Z^{\mathbf{a}^{1} \mathbf{b}^{1}} \X^{a^{1}_x} \Z^{a^{1}_z} \ket{\psi_1} \, .
\]
Knowing the relevant variables from his measurement outcomes in the previous steps, Bob performs the operation $  \X^{\hat{b}^1_x}\Z^{\hat{b}^1_{z}} (\P^{\mathbf{b}^{1}}_{w_1})^{\dagger}$ to transform the state to $
 \Z^{\mathbf{a}^{1} \mathbf{b}^{1}} \X^{a^{1}_x} \Z^{a^{1}_z} \ket{\psi_1}$.
\item[Step 1.d] For this step the players share two sets of $n$ EPR pairs, 
one set labeled ``$\mathbf{b}^{1}=0$'', the other set labeled ``$\mathbf{b}^{1}=1$''.
Bob teleports the state to Alice using the set corresponding to the
value of $\mathbf{b}^{1}$, with teleportation outcomes $b^{2}_x$ and $b^{2}_z$.

\item[Step 1.e] The set of qubits corresponding to the correct value of $\mathbf{b}^{1}$ are in the state
\[
\X^{b^{2}_x} \Z^{b^{2}_x} \Z^{\mathbf{a}^{1} \mathbf{b}^{1}} \X^{a^{1}_x} \Z^{a^{1}_z} \ket{\psi_1} \, .
\]
On the set labeled  ``$\mathbf{b}^{1}=0$'', Alice applies $\X^{a^{1}_x} \Z^{a^{1}_z}$,
and on the set labeled ``$\mathbf{b}^{1}=1$'' Alice applies $\X^{a^{1}_x} \Z^{a^{1}_z} \Z^{\mathbf{a}^{1}}_{w_1}$, so that the state (at the correct set of qubits) equals
$
\X^{b^{2}_{x}} \Z^{b^{2}_{z}} \ket{\psi_1}
$.

We are now in almost the same situation as before the first step: Alice is in possession of a state for which Bob completely knows the needed teleportation corrections -- with the difference that Alice does not know which of the two sets that is.

\item[Steps 2$\,\mathbf{ \dots k}$] The players repeat the protocol from Step 1, but Alice performs all steps in parallel for \emph{all} sets of states. The needed resources then double with each step: two sets for step 2, four
for step 3, etc.

\item[Step k+1, final step] When having executed this protocol for the entire circuit, Alice
only teleports Bob's qubits back to him, i.e.\ the qubits corresponding to the last $n/2$ wires, instead of the entire state, so that
in the correct groups, Alice and Bob are in possession of the state $\ket{\psi_k}$ up to simple teleportation corrections. Then, in their step of simultaneous communication, the players exchange all teleportation measurement outcomes.
After receiving these measurement outcomes, the players discard the qubits that did not contain the state, and perform the Pauli corrections on the correct qubits.
\end{description}

The needed EPR pairs for this protocol consist of $n/2$ for Step 0. Then
every set uses at most $3n$ pairs: $n$ for the teleportation of Alice to Bob, and $2n$ for the teleportation back.
The $t$-th step of the circuit starts with $2^{t-1}$ sets of parallel executions,
therefore the total entanglement is upper bounded by $n/2 + \sum^{k}_{t=1}{2^{t-1} 3n} \leq 3 n 2^k$.
\end{proof}

\section{Conditional application of phase gate using garden-hose protocols}\label{sec:nonlocalphase}
The following lemma connects the difficulty of removing an unwanted phase gate that is applied conditional on a function $f$, to the garden-hose complexity of  $f$.
This lemma is the main technical tool which we use to non-locally compute quantum circuits with a dependence on the $\T$-depth.
\begin{lemma}\label{lem:nonlocalphase}
Assume Alice has a single qubit with state $\P^{f(x,y)} \ket{\psi}$, for
binary strings $x,y\in \{0,1\}^n$, where Alice knows the string $x$ and Bob knows $y$.
Let $\gh(f)$ be the garden-hose complexity of the function $f$. The following two statements hold: 
 
\begin{enumerate}
\item There exists an instantaneous protocol without any communication which uses
$2 \gh(f)$ pre-shared EPR pairs after which a
known qubit of Alice is in the state
$\X^{g(\hat{x},\hat{y})} \Y^{h(\hat{x},\hat{y})} \ket{\psi}$. Here
$\hat{x}$ depends only on $x$ and the $2\gh(f)$ bits that describe the
measurement outcomes of Alice, and $\hat{y}$ depends on $y$ and
the measurement outcomes of Bob.
\item The garden-hose complexities of the functions $g$ and $h$ are at most linear in the complexity of the function $f$. More precisely, $\gh(g) \leq 4 \gh(f) + 1$ and $\gh(h) \leq 11 \gh(f) + 2$.
\end{enumerate}
\end{lemma}

\begin{proof}
To prove the first statement we will construct a quantum protocol
that uses $2 \gh(f)$ EPR pairs,
which is able to remove the conditional phase gate. The quantum protocol
uses the garden-hose protocol for $f$ as a black box.

For the second part of the statement of the lemma, we construct garden-hose
protocols which are able to compute the teleportation corrections
that were incurred by executing our quantum protocol. By
explicitly exhibiting these protocols, we give an upper bound
to the garden-hose complexity of the X correction $g$ and the
Z correction $h$.

\begin{figure}[h]
\begin{center}
 \begin{tikzpicture}
  \coordinate (psi) at (1, 10);
 \coordinate (tlbox1) at (1,9);
 \coordinate (brbox1) at (7,6);
 \coordinate (tlbox2) at (1,5);
 \coordinate (brbox2) at (7,2);
 
\coordinate (pipeA1) at (1,9);
\coordinate (pipeB1) at (7,9);
\coordinate (pipeA2) at (1,8);
\coordinate (pipeB2) at (7,8);
\coordinate (pipeA3) at (1,7);
\coordinate (pipeB3) at (7,7);

\coordinate (pipeA4) at (1,5);
\coordinate (pipeB4) at (7,5);
\coordinate (pipeA5) at (1,4);
\coordinate (pipeB5) at (7,4);
\coordinate (pipeA6) at (1,3);
\coordinate (pipeB6) at (7,3);

\coordinate (psiout) at (0.5,1.5);
 
\draw[fill=white] (pipeA1) circle (0.1cm);
\draw[fill=white] (pipeB1) circle (0.1cm);
\draw[thick, small snake] ($ (pipeA1) + (0.2,0) $) -- ($ (pipeB1) - (0.2,0) $);

\draw[fill=white] (pipeA2) circle (0.1cm);
\draw[fill=white] (pipeB2) circle (0.1cm);
\draw[thick, small snake] ($ (pipeA2) + (0.2,0) $) -- ($ (pipeB2) - (0.2,0) $);

\draw[fill=white] (pipeA3) circle (0.1cm);
\draw[fill=white] (pipeB3) circle (0.1cm);
\draw[thick, small snake] ($ (pipeA3) + (0.2,0) $) -- ($ (pipeB3) - (0.2,0) $);

\draw[fill=white] (pipeA4) circle (0.1cm);
\draw[fill=white] (pipeB4) circle (0.1cm);
\draw[thick, small snake] ($ (pipeA4) + (0.2,0) $) -- ($ (pipeB4) - (0.2,0) $);

\draw[fill=white] (pipeA5) circle (0.1cm);
\draw[fill=white] (pipeB5) circle (0.1cm);
\draw[thick, small snake] ($ (pipeA5) + (0.2,0) $) -- ($ (pipeB5) - (0.2,0) $);

\draw[fill=white] (pipeA6) circle (0.1cm);
\draw[fill=white] (pipeB6) circle (0.1cm);
\draw[thick, small snake] ($ (pipeA6) + (0.2,0) $) -- ($ (pipeB6) - (0.2,0) $);

\draw[fill=white, fill opacity=1] (psi) circle (0.1cm) node[anchor=west] {$\P^{f(x,y)} \ket{\psi}$};
\node at (4,8) [draw, thick, minimum width=5cm, minimum height=3cm, fill=white, fill opacity=1, align=center] {Teleport according to \\GH protocol for $f$};

\node at ($ (pipeB1) + (0.6, 0) $)  [draw, thick] {$\P^{-1}$};
\node at ($ (pipeB2) + (0.6, 0) $)  [draw, thick] {$\P^{-1}$};
\node at ($ (pipeB3) + (0.6, 0) $)  [draw, thick] {$\P^{-1}$};

\draw ($ (psi) - (0.3, 0) $) to[out=200, in=160] ($ (pipeA1) - (0.3, 0) $);
\draw ($ (pipeA2) - (0.3, 0) $) to[out=200, in=160] ($ (pipeA5) - (0.3, 0) $);
\draw ($ (pipeA3) - (0.3, 0) $) to[out=200, in=160] ($ (pipeA6) - (0.3, 0) $);

\draw ($ (pipeB1) + (1.1, 0) $) to[out=-30, in=20] ($ (pipeB4) + (0.3, 0) $);
\draw ($ (pipeB2) + (1.1, 0) $) to[out=-30, in=20] ($ (pipeB5) + (0.3, 0) $);
\draw ($ (pipeB3) + (1.1, 0) $) to[out=-30, in=20] ($ (pipeB6) + (0.3, 0) $);

\draw[dashed,<-] ($ (pipeA4) - (0.3, 0) $) to[out=200, in=100] ($ (psiout) + (0, 0.3) $);
\node[anchor=west] at (psiout)  {$\X^{g(\hat{x},\hat{y})} \Z^{h(\hat{x},\hat{y})} \ket{\psi}$};
 
 \draw[fill=white] (psi) circle (0.1cm) node[anchor=west] {$\P^{f(x,y)} \ket{\psi}$};
\node at (4,4) [draw, thick, minimum width=5cm, minimum height=3cm, fill=white, fill opacity=1, align=center] {GH protocol for $f$ (copy)};

 \end{tikzpicture}
 \end{center}
\caption{Schematic overview of the quantum protocol to undo the conditionally-present phase gate on $\ket{\psi}$. The solid connections correspond to Bell measurements.}
\label{fig:nonlocalphase}
\end{figure}
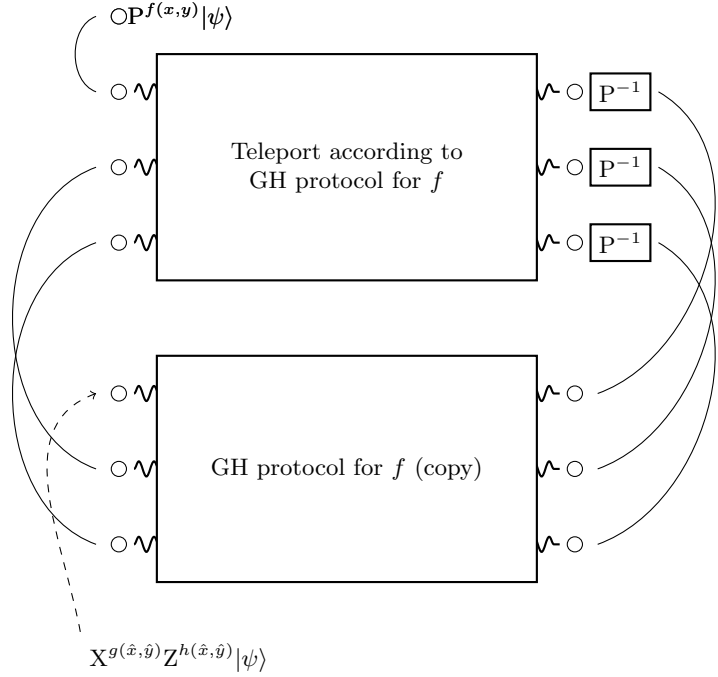

The quantum protocol is shown as Figure~\ref{fig:nonlocalphase}.
Alice and Bob execute the garden-hose protocol with the state $\P^{f(x,y)} \ket{\psi}$, i.e.\ they teleport
the state back and forth, with the EPR pairs chosen depending on $x$ and $y$.
Afterwards, if $f(x,y)=0$, the qubit will be at one of the unmeasured EPR halves on Alice's side, and if $f(x,y)=1$ the qubit will be on 
Bob's side. The state of the qubit will be $\X^{g'(x', y')} \, \Z^{h'(x', y')} \, \P^{f(x,y)} \ket{\psi} = \P^{f(x,y)} \X^{g'(x', y')} \, \Z^{h'(x', y') \oplus f(x,y) g'(x',y')} \,  \ket{\psi}$,
for some functions $g'$ and $h'$.

On each qubit on Bob's side, corresponding with an `open pipe' in
the garden-hose model, Bob applies $\P^{-1}$, so that the state of the qubit is now equal to $\X^{g'(x', y')} \, \Z^{h'(x', y') \oplus f(x,y) g'(x',y')} \, \ket{\psi}$. 
The exact location
of our qubit depends on the protocol, and is unknown to both players. Here $x'$ and $y'$ are the measurement outcomes of Alice and Bob in this first half of the protocol.

To return the qubit to a known position without an extra communication step, we employ a trick that uses the reversibility of the
garden-hose model. Alice and Bob repeat the exact same garden-hose strategy, except they leave the start open, and connect the open ends between
the original and the copy. Alice performs a Bell measurement between the first open qubit in the original, and the first open qubit in the copy, etc.
Bob does the same, after he applied the $\P$ gates. Afterwards, the qubit will be present in the start location, `water tap' in garden-hose terminology,
of the copied game, since it has followed the exact same path backwards. The final state of the qubit now is $\X^{g(\hat{x}, \hat{y})} \, \Z^{h(\hat{x}, \hat{y})} \,  \ket{\psi}$, for some functions
$g$ and $h$ and $\hat{x}$ and $\hat{y}$ the measurement outcomes of Alice and Bob respectively. The total entanglement consumption is $2 \gh(f)$.

Every measurement corresponds to a connection of two pipes in the
garden-hose model, therefore each player performs at most $\gh(f)$ teleportation
measurements, of which the outcomes can be described by $2 \gh(f)$ bits.

Label the EPR pairs with numbers from $\{1, 2, \dots, 2\gh(f)\}$, and use the label 0 for the register holding the starting qubit $\ket{\psi}$.
Let $\mathcal{A}$ be a list of disjoint pairs of the indices of the EPR pairs that Alice uses for teleportation in this protocol, and let $a_x, a_z \in \{0,1\}^{|\mathcal{A}|}$ be the bit strings that
respectively hold the X and Z outcomes of the corresponding Bell measurements. Similarly, let $\mathcal{B}$ be a list of the indices of the EPR pairs that Bob uses, and let $b_x, b_z \in \{0,1\}^{|\mathcal{B}|}$ be the bit strings that
hold the measured X and Z corrections. 

To show the second part of the statement, we will construct a garden-hose protocol which tracks the newly-incurred Pauli corrections from teleporting the qubit back-and-forth, by following the qubit through the path
defined by $\mathcal{A}$ and~$\mathcal{B}$.

We will first construct the protocol for the final X-correction, a function we denoted by $g$. The protocol is also schematically shown as Figure~\ref{fig:gh-x}. Note that to compute the X correction the conditional presence of the phase gate is not important:
independent of whether $f(x,y)$ equals 1 or 0, we only need to track the X teleportation corrections that the qubit incurred by being teleported back-and-forth by Alice and Bob.
An efficient garden-hose protocol for $g$ is given by the following.

Use two pipes for each EPR pair in the protocol, $2 \gh(f)$ pairs of 2 pipes each. Label the top pipe of some pair $i$ by $I_i$, and the bottom pipe by $X_i$.
We will iterate over all elements of $\mathcal{A}$, i.e.\ all performed Bell measurements by Alice. Consider some element of $\mathcal{A}$, say the $k$-th pair $\mathcal{A}_k$ which consists of $\{i,j\}$. If the corresponding correction $b_{x,k}$ equals 0, we
connect the pipe labeled $I_i$ with the pipe labeled $I_j$ and the pipe labeled $X_i$ with the pipe labeled $X_j$. Otherwise, if $b_{x,k}$ equals 1, we connect them crosswise, so we connect $I_i$ with $X_j$ and $X_i$ with $I_j$.
Finally, the place where the qubit ends up after the protocol is unique (and is the only unmeasured qubit out of all $2 \gh(f)$ EPR pairs). For the set of open pipes corresponding to that EPR pair, say number $i^*$, we use one extra pipe to which 
we connect $X_{i^*}$, so that the water ends up at Bob's side for the 1-output.
This garden-hose protocol computes the X correction on the qubit, and uses $4 \gh(f) + 1$ pipes in total, therefore $\gh(g) \leq 4\gh(f) + 1$.

\begin{figure}[h]
 \begin{center}
\begin{tikzpicture}

\coordinate (psi) at (3, 11);
\coordinate (eprA1) at (3,10);
\coordinate (eprB1) at (9,10);
\coordinate (eprA2) at (3,9);
\coordinate (eprB2) at (9,9);
\coordinate (eprA3) at (3,8);
\coordinate (eprB3) at (9,8);

\coordinate (tap) at (3,6.4);

\coordinate (pipeA1) at (3,5.6);
\coordinate (pipeB1) at (9,5.6);
\coordinate (pipeA1X) at (3, 5);
\coordinate (pipeB1X) at (9, 5);

\coordinate (pipeA2) at (3, 4.2);
\coordinate (pipeB2) at (9, 4.2);
\coordinate (pipeA2X) at (3, 3.6);
\coordinate (pipeB2X) at (9, 3.6);

\coordinate (pipeA3) at (3, 2.8);
\coordinate (pipeB3) at (9, 2.8);
\coordinate (pipeA3X) at (3, 2.2);
\coordinate (pipeB3X) at (9, 2.2);

\draw[fill=white] (psi) circle (0.1cm) node[anchor=west] {$\ket{\psi}$};

\draw[fill=white] (eprA1) circle (0.1cm) node[anchor=east] {};
\draw[fill=white] (eprB1) circle (0.1cm) node[anchor=east] {};
\draw[thick, snake it] ($ (eprA1) + (0.5,0) $) -- node[above]{ {\small EPR pair 1} } ($ (eprB1) - (0.5,0) $);

\draw[fill=white] (eprA2) circle (0.1cm) node[anchor=east] {};
\draw[fill=white] (eprB2) circle (0.1cm) node[anchor=east] {};
\draw[thick, snake it] ($ (eprA2) + (0.5,0) $) -- node[above]{ {\small EPR pair 2} } ($ (eprB2) - (0.5,0) $);

\draw[fill=white] (eprA3) circle (0.1cm) node[anchor=east] {};
\draw[fill=white] (eprB3) circle (0.1cm) node[anchor=east] {};
\draw[thick, snake it] ($ (eprA3) + (0.5,0) $) -- node[above]{ {\small EPR pair 3} } ($ (eprB3) - (0.5,0) $);

\draw ($ (psi) - (0.3,0) $) to[out=200, in=160] node[left] {$a_{x,1}, a_{z,1}$ } ($ (eprA1) - (0.3,0) $);
\draw ($ (eprA2) - (0.3,0) $) to[out=200, in=160] node[left] {$a_{x,2}, a_{z,2}$ } ($ (eprA3) - (0.3,0) $);
\draw ($ (eprB1) + (0.3,0) $) to[out=-20, in=20] node[right] {$b_{x,1}, b_{z,1}$ } ($ (eprB2) + (0.3,0) $);

\node[draw, style=single arrow, fill=lightgray, rotate=-90] at (6,7.2) { $\quad\,$ };

\draw[fill=black] (tap) circle (0.1cm) node[anchor=west] {\hspace{0.2cm}\emph{tap}};

\draw[fill=black] (pipeA1) circle (0.05cm);
\draw[fill=black] (pipeB1) circle (0.05cm);
\draw[very thick] (pipeA1)  -- node[above] {$I_1$} (pipeB1);

\draw[fill=black] (pipeA1X) circle (0.05cm);
\draw[fill=black] (pipeB1X) circle (0.05cm);
\draw[very thick] (pipeA1X)  -- node[above] {$X_1$} (pipeB1X);

\draw[fill=black] (pipeA2) circle (0.05cm);
\draw[fill=black] (pipeB2) circle (0.05cm);
\draw[very thick] (pipeA2)  -- node[above] {$I_2$} (pipeB2);

\draw[fill=black] (pipeA2X) circle (0.05cm);
\draw[fill=black] (pipeB2X) circle (0.05cm);
\draw[very thick] (pipeA2X)  -- node[above] {$X_2$} (pipeB2X);

\draw[fill=black] (pipeA3) circle (0.05cm);
\draw[fill=black] (pipeB3) circle (0.05cm) node[right]{{\small(out if $a_{x,1} \oplus b_{x,1} \oplus a_{x,2} = 0$)}};
\draw[very thick]  (pipeA3) -- node[above] {$I_3$} (pipeB3);

\draw[fill=black] (pipeA3X) circle (0.05cm);
\draw[fill=black] (pipeB3X) circle (0.05cm) node[right]{{\small(out if $a_{x,1} \oplus b_{x,1} \oplus a_{x,2} = 1$)}};
\draw[very thick] (pipeA3X) -- node[above] {$X_3$} (pipeB3X);

\draw[thick] ($ (tap) - (1.3,0) $) to[out=200, in=160] ($ (pipeA1) - (1.3,0) $);
\node at ($ (tap) - (1.5,0) + (0,0.2) $) {\scriptsize{$a_{x,1} \!=\! 0$} };

\draw[thick]  ($ (tap) - (0.3,0) $) to[out=200, in=160] ($ (pipeA1X) - (0.3,0) $);
\node at ($ (tap) - (0.35,0) + (0,0.2) $) {\scriptsize{$a_{x,1} \!=\! 1$} };

\draw[thick] ($ (pipeA2) - (1.3,0) $) to[out=200, in=160] ($ (pipeA3) - (1.3,0) $);
\draw[thick] ($ (pipeA2X) - (1.3,0) $) to[out=200, in=160] ($ (pipeA3X) - (1.3,0) $);
\node at ($ (pipeA2) - (1.5,0) + (0,0.2) $) {\scriptsize{$a_{x,2} \!=\! 0$} };

\draw[thick] ($ (pipeA2) - (0.3,0) $) to[out=200, in=160] ($ (pipeA3X) - (0.3,0) $);
\draw[thick] ($ (pipeA2X) - (0.3,0) $) to[out=200, in=160] ($ (pipeA3) - (0.3,0) $);
\node at ($ (pipeA2) - (0.35,0) + (0,0.2) $) {\scriptsize{$a_{x,2} \!=\! 1$} };

\draw[thick] ($ (pipeB1) + (0.3,0) $) to[out=-20, in=20] ($ (pipeB2) + (0.3,0) $);
\draw[thick] ($ (pipeB1X) + (0.3,0) $) to[out=-20, in=20] ($ (pipeB2X) + (0.3,0) $);
\node at ($ (pipeB1) + (0.35,0) + (0,0.2) $) {\scriptsize{$b_{x,1} \!=\! 0$} };

\draw[thick] ($ (pipeB1) + (1.3,0) $) to[out=-20, in=20]  ($ (pipeB2X) + (1.3,0) $);
\draw[thick] ($ (pipeB1X) + (1.3,0) $) to[out=-20, in=20] ($ (pipeB2) + (1.3,0) $);
\node at ($ (pipeB1) + (1.5,0) + (0,0.2) $) {\scriptsize{$b_{x,1} \!=\! 1$} };

\end{tikzpicture}
 \end{center}
\caption{Example garden-hose protocol to compute the Pauli X incurred by Alice and Bob teleporting a qubit back-and-forth. When a teleportation requires a Pauli X correction, the corresponding
pipes are connected crosswise, and otherwise they are connected in parallel.}
\label{fig:gh-x}
\end{figure}
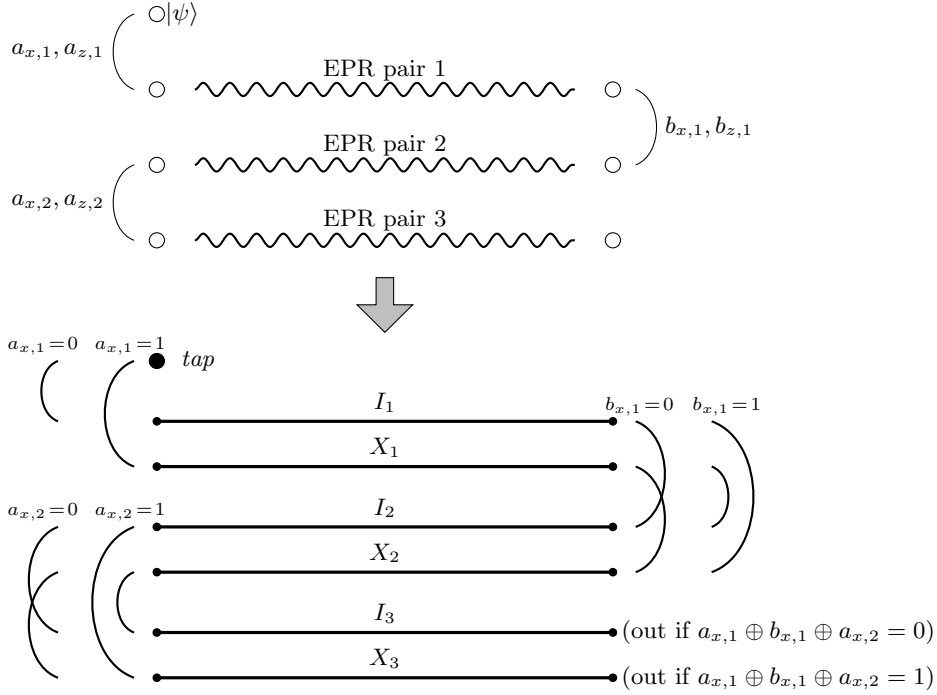

For the Z-correction we can build a garden-hose protocol using the same idea, but there is one complication we have to take care of. At the start of the protocol, there might be an unwanted phase gate present on the state.
If some teleportation is performed before
this phase gate is corrected, say by Alice with outcomes $a_x, a_z$, then the effective correction can be written as $\X^{a_x} \Z^{a_z} \P = \P \X^{a_x} \Z^{a_x \oplus a_z}$.
That is, for the part of the protocol that the unwanted phase gate is present, a Bell measurement gives a Z-correction whenever the \emph{exclusive or} of the X- and Z-outcomes is 1, instead of just
when the Z-outcome is 1. We will therefore use the garden-hose protocol that computes whether $f(x,y)=1$, that is, compute whether the phase gate is present, and then
execute a slightly different garden-hose protocol for each case. 

\begin{figure}[ht]
\begin{center}
\begin{tikzpicture}

\coordinate (tap) at (0.8,9.5);
\coordinate (box1t) at (1,10);
\coordinate (box1b) at (7,7);

\coordinate (box2t) at (1,6);
\coordinate (box2b) at (7,4);

\coordinate (box3t) at (1,3);
\coordinate (box3b) at (7,0);
\coordinate (box3ml) at ($(box3t) - (0, 1.5)$);
\coordinate (box3mr) at ($(box3b) + (0, 1.5)$);

\draw[thick] (box1t) rectangle node[align=center, thick] {Unique-output GH protocol \\ for $f(x,y)$ \\ \\(Lemma~\ref{lem:ghsingleoutput})} (box1b);
\draw[thick] (box2t) rectangle node[align=center, thick] {Compute correction
using\\Z outcomes} (box2b);
\draw[thick] (box3t) rectangle (box3b);
\draw[draw=none] (box3t) rectangle node[align=center, thick] {Compute correction using \\ X$\oplus$Z outcomes of first part} (box3mr);
\draw[draw=none] (box3ml) rectangle node[align=center, thick] {Compute correction using \\ Z outcomes of the rest} (box3b);
\draw[thick, dashed] (box3ml) -- (box3mr);

\draw[fill=black] ($(tap) - (0.2,0)$) circle (0.1cm) node[anchor=east] {\emph{tap }\hspace{0.4cm}};

\draw[thick]  ($ (tap) - (0, 0) $) to ($ (box1t) - (0, 0.5) $);
\draw[thick]  ($ (box1t) - (0, 1.5) $) to[out=200, in=160] node[anchor=east, fill=white, fill opacity=1] {$f(x,y)=0$} ($ (box2t) - (0, 0.5) $);
\draw[thick]  ($ (box1t) - (0, 2.5) $) to[out=200, in=160] node[anchor=east, fill=white, fill opacity=1] {$f(x,y)=1$} ($ (box3t) - (0, 0.5) $);

\draw[fill=black, thick]  ($ (box2t) - (0.3, 1.5)$) circle(0.05cm) node[anchor=east] {0} -- ($ (box2t) - (0, 1.5)$);
\draw[fill=black, thick]  ($ (box2t) + (6.3, -1.5)$) circle(0.05cm) node[anchor=west] {1} -- ($ (box2t) + (6, -1.5)$);

\draw[fill=black, thick]  ($ (box3t) - (0.3, 2.5)$) circle(0.05cm) node[anchor=east] {0} -- ($ (box3t) - (0, 2.5)$);
\draw[fill=black, thick]  ($ (box3t) + (6.3, -2.5)$) circle(0.05cm) node[anchor=west] {1} -- ($ (box3t) + (6, -2.5)$);

\end{tikzpicture}
\end{center}
\caption{Sketch of garden-hose protocol for the Z correction. The bottom two boxes use the construction which was used for the X-correction; in the top case using the Z-outcomes for all measurements, in the bottom case using the
parity of the X- and Z-outcomes for those teleportations that happened before the unwanted phase gate was removed.}
\label{fig:gh-z}
\end{figure}
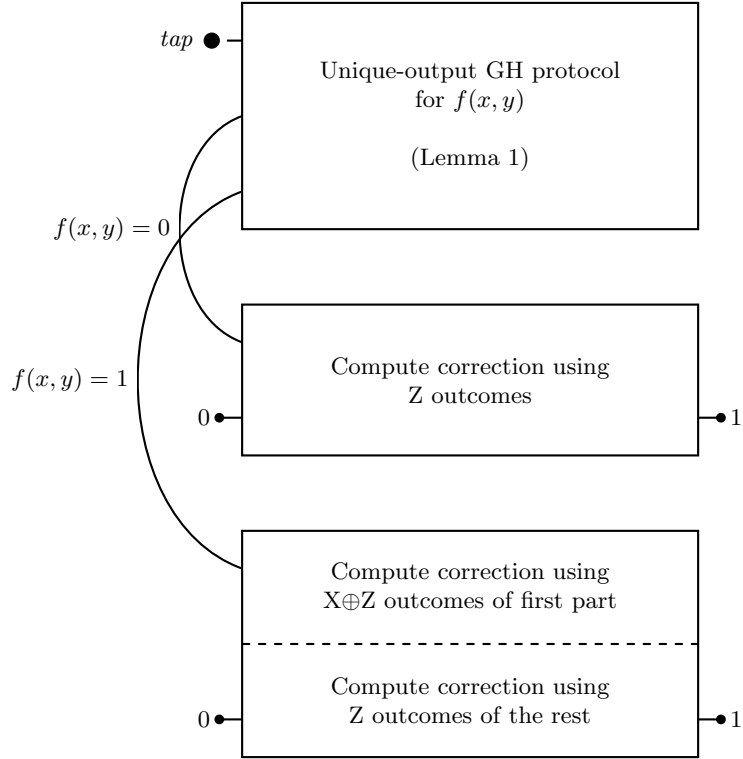

See Figure~\ref{fig:gh-z} for an overview of the different parts of this garden-hose protocol for the Z-correction $h$.
Using Lemma~\ref{lem:ghsingleoutput} we can transform the garden-hose protocol for $f$ into a garden-hose protocol for $f$ with unique 0 and 1 outputs at Alice's side, of size $3 \gh(f)$.\footnote{If the unique 0 output has to be at Alice's side, and
the unique 1 output at Bob's side, the construction uses $3\gh(f)+1$ pipes.
It is an easy exercise to show that the construction of Lemma~\ref{lem:ghsingleoutput} needs one pipe less if Alice
wants to have both the designated 0 output and the 1 output.}  For the 0 output, that is if there was no unwanted phase gate present, we can track the Z corrections in exactly the same way as we did for the X corrections, for a subprotocol of size $4\gh(f)+1$.
For the 1 output there was in fact a phase gate present, for the teleportations that happened in the protocol before the $\P^{-1}$ corrections. For that part of the protocol, we execute the correction-tracking protocol using the XOR of the X- and Z-measurement
outcomes. For all teleportations after the phase correction, we again track the correction using just the Z-outcomes, since there is no phase gate present anymore. This part of the garden-hose protocol also uses $4\gh(f)+1$ pipes, for a total of $11 \gh(f) + 2$.
\end{proof}

\section{Low T-depth quantum circuits}\label{sec:tdepth}\index{T-depth}
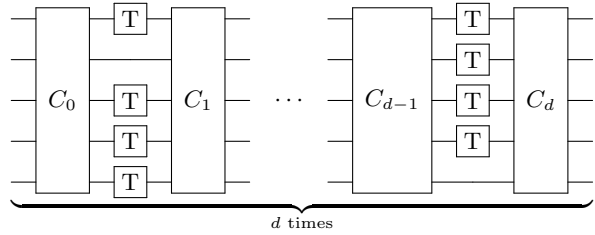
\begin{figure}[ht]
\begin{center}
$
\underbrace{
\Qcircuit @C=1em @R=.4em {
& \multigate{4}{C_0} & \gate{\T} & \multigate{4}{C_1} & \qw & & & \multigate{4}{C_{d-1}} & \gate{\T} & \multigate{4}{C_{d}} & \qw \\
& \ghost{C_0} 		 & \qw & \ghost{C_1} 		  & \qw & & & \ghost{C_{d-1}} & \gate{\T} & \ghost{C_{d}} & \qw\\
& \ghost{C_0} 		 & \gate{\T} & \ghost{C_1} 		  & \qw & \push{\cdots} & & \ghost{C_{d-1}} & \gate{\T} & \ghost{C_{d}} & \qw\\
& \ghost{C_0} 		 & \gate{\T} & \ghost{C_1} 		  & \qw & & & \ghost{C_{d-1}} & \gate{\T} & \ghost{C_{d}} & \qw\\
& \ghost{C_0} 		 & \gate{\T} & \ghost{C_1} 		  & \qw & & & \ghost{C_{d-1}} & \qw & \ghost{C_{d}} & \qw\\
}
}_{\text{$d$ times}}
$\end{center}
\caption{An example circuit with T-depth $d$. The $C_i$ gates represent subcircuits consisting only of operations from the Clifford group $\C$. A layer does not necessarily have a $\T$ gate on all wires.}
\label{fig:tdepth}
\end{figure}

\begin{theorem}\label{thm:tdepth}
Let $C$ be an $n$-qubit quantum circuit with gates out of the Clifford+T gate set, where $C$ has T-depth $d$. Then there exists a protocol for two-party instantaneous non-local computation of $C$, where each party receives $n/2$ qubits, which uses a pre-shared entangled state of $O(\, (68n)^{d}\, )$ EPR pairs.
That is, $\INQC(C) \leq O(\, (68n)^{d}\, )$.
\end{theorem}
\begin{proof}
As in the proof of Theorem~\ref{thm:tcount}, we write the input state $\ket{\psi}$, and write the correct quantum state after step $t$ of the circuit as $\ket{\psi_t}$.
At a step $t$, the circuit alternates between a \emph{layer} of T gates\footnote{We will assume that for each layer of T gates \emph{all} wires have a T gate.
This is only done to avoid introducing extra notation needed when instead the gates are only applied to a subset -- the protocol easily generalizes to the more common general situation.} and a subcircuit consisting of only Clifford gates, $C_t$.

The high-level idea of this protocol is as follows. During steps 1 to $t$, Alice will hold the entire
uncorrected state and performs a layer of the circuit: 
she performs a layer of T gates and then a Clifford subcircuit. 
The Pauli corrections at each step are a function of earlier
teleportation outcomes of both Alice and Bob.
These functions determine for
each qubit whether that qubit now has obtained an unwanted extra $\P$ gate when Alice performs the layer of T gates.
The players then, for each qubit, correct this extra gate using Lemma~\ref{lem:nonlocalphase} --
removing the unwanted phase gate from the qubit in a way that both players still know its location. 

At each step we express the corrections as functions of earlier measurements
and consider their garden-hose complexity, which is important when
using Lemma~\ref{lem:nonlocalphase}. The Clifford subcircuit
takes the correction functions to the XOR of several earlier functions. We can bound
the growth in garden-hose complexity by taking XORs using Lemma~\ref{lem:sumgh}.
Taken together, the garden-hose complexity grows with a factor of at most a constant times $n$ each step.

We will use $f^t_{x,i}$ to denote the function that describes the presence of an X correction on qubit $i$, at step $t$ of the protocol. Similarly, $f^t_{z,i}$ is the function that describes the Z correction on qubit $i$
at step $t$.
Both will always be functions of outcomes of earlier teleportation
measurements of Alice and Bob.
For any $t$, let $m_t$ be the maximum garden-hose complexity over all the key functions at step $t$. 

\begin{description}
\item[Step 0] Bob teleports his qubits, the qubits labeled $n/2$ up to $n$, to Alice, obtaining the measurement outcomes $b^{0}_{x,1}, \dots, b^{0}_{x,n/2}$ and $b^{0}_{z,1}, \dots, b^{0}_{z,n/2}$.
On these uncorrected qubits, Alice executes the Clifford subcircuit $C_0$.

Then, since Bob also knows how $C_0$ transforms the keys,
the functions describing the Pauli corrections can all either be described
by a single bit of information which is locally computable by Bob, or are constant and therefore
known by both players. Let $f^{0}_{x,i}$ and $f^{0}_{z,i}$ be the resulting key function for any qubit $i$.
The garden-hose
complexity of all these key functions is constant: $\gh(f^{0}_{x,i}) \leq 3$ and $\gh(f^{0}_{z,i}) \leq 3$,
and therefore also for the maximum garden-hose complexity we have $m_0 \leq 3$. 

\item[Step $t=1,\dots,d$]

At the start of the step, the X and Z corrections on any wire $i$ are given by $f^{t-1}_{x,i}$ and $f^{t-1}_{z,i}$ respectively.

Alice applies the $\T$ gates on all wires. Any wire $i$ now
has an unwanted $\P$ if and only if 
$f^{t}_{x,i}$ equals 1.

Alice and Bob apply the construction of Lemma~\ref{lem:nonlocalphase},
which removes this unwanted phase gate. Let $g^t_i$ be the function
describing the extra X correction incurred by this protocol, so
that the new X correction can be written as $f^{t}_{x,i} \oplus g^t_i$. Let
$h^t_i$ be the function describing the Z correction, so that 
the total Z correction is $f^{t}_{z,i} \oplus h^t_i$.
The entanglement cost of this protocol is given by $2\gh(f^{t}_{x,i})$
and the garden-hose complexities of the new functions are at most
$\gh(g^t_i) \leq 4 \gh(f^{t}_{x,i}) + 1$ and 
$\gh(h^t_i) \leq 11 \gh(f^{t}_{x,i}) + 2$.

Alice now executes the Clifford subcircuit $C_t$. 
The circuit $C_t$ determines how the current Pauli corrections, i.e.\ the key functions, transform.  For a specification
of the possible transformations, see Section~\ref{sec:paulitransform}. These new keys 
are formed by taking the exclusive OR of some subset of keys that were present in the previous step\footnote{This is slightly more general than necessary, since not all possible key transformations of this form are actually possible -- only those transformations
generated by the possibilities in Section~\ref{sec:paulitransform} can occur.}.

Consider the worst case key for our construction: a key which is given by the XOR of all keys that were present when the Clifford subcircuit was executed. 
Applying Lemma~\ref{lem:sumgh}, the worst-case key function of the form $\bigoplus_{i=1}^n f^{t-1}_{x,i} \oplus g^{t}_{i} \oplus f^{t-1}_{z,i} \oplus h^{t}_{i}$ has garden-hose complexity at most
\begin{align}
m_t & \leq 4\left(\sum_{i=1}^n \gh(f^{t-1}_{x,i}) + \gh(g^{t}_{i}) + \gh(f^{t-1}_{z,i}) +\gh( h^{t}_{i})\right) + 1 \nonumber \\
& \leq 4\left(\sum_{i=1}^n \gh(f^{t-1}_{x,i}) + 4 \gh(f^{t-1}_{x,i}) + 1 +  \gh(f^{t-1}_{z,i}) +11 \gh(f^{t-1}_{x,i}) + 2\right) + 1 \nonumber\\
& \leq 4\left(\sum_{i=1}^n m_{t-1} + 4 m_{t-1} + 1 + m_{t-1} +11 m_{t-1} + 2\right) + 1 \nonumber \\
& = 68 n m_{t-1} + 12 n + 1 \label{eq:mt-recurrence} \,.
\end{align}

\item[Step $d+1$, final step] Alice teleports the last $n/2$ qubits back to Bob.
Alice and Bob exchange all results of teleportation measurements and locally perform the needed corrections, using both players' measurement outcomes.
\end{description}

At every step $t$, the protocol uses at most $2 n m_{t-1}$ EPR pairs for the protocol
which corrects the phase gate. Using that $m_0 \leq 3$, we can write the upper bound of Equation~\ref{eq:mt-recurrence} as 
the closed form $m_t \leq c_1 (68n)^t + c_2$, with $c_1 = \frac{216n-2}{68n-1} \approx \frac{54}{17}$ and $c_2 = 3-\frac{216n-2}{68n-1} \approx -\frac{3}{17}$. The total entanglement use therefore is bounded by $
 \sum_{t=1}^d 2 n m_{t-1} \leq O(\, (68n)^{d}\, )$.
\end{proof}

\section{The Interleaved Product protocol}\label{sec:interleavedattack}
Chakraborty and Leverrier~\cite{CL15} recently proposed a scheme for quantum position verification based on the interleaved multiplication of unitaries, the \emph{Interleaved Product protocol}, 
denoted by $\GIP(n,t,\etaerr, \etaloss)$. The parameter $n$ concerns the number of qubits that are involved in the protocol in parallel, while $t$ scales with the amount of classical information that the protocol uses.
Their paper analyzed several different attacks on this scheme, which all required exponential entanglement in the parameter $t$.
In this section, as an application of the proof strategy of Theorem~\ref{thm:tdepth}, we present an attack on the Interleaved Product protocol which requires entanglement polynomial in $t$.

The original protocol is described in terms of the actions of hypothetical honest parties and also involves checking of timings at spatial locations.
For simplicity, we instead only describe a two-player game, for players Alice and Bob, such that a high probability of winning this game suffices to break the scheme.
Let $x$ be a string $x \in_R \{0,1\}^n$, and let $U$ be a random (single-qubit) unitary operation, i.e.\ a random element of $U(2)$. Alice receives $t$ unitaries $(u_i)_{i=1}^t$, and Bob receives $t$ unitaries $(v_i)_{i=1}^t$
such that $U = \prod_{i=1}^t u_i v_i$.
Alice receives the state $U^{\otimes n} \ket{x}$. The players are allowed one round of simultaneous communication.
To break the protocol $\GIP(n,t,\etaerr, \etaloss)$, after the round of simultaneous communication the players need to output an identical string $y \in \{\emptyset, 0,1 \}^n$ such that 
the number of bits where $y$ is different from $x$ is at most $\etaerr n$ and the number of empty results $\emptyset$ is at most $\etaloss n$.
We will consider attacks on the strongest version of the protocol, where we take $\etaloss=0$.

\begin{theorem}\label{thm:attackipp}
There exists an attack on $\GIP(n, t, \etaerr, \etaloss=0)$ that requires $p(t/{\etaerr})$ EPR pairs per qubit of the protocol,
for some polynomial $p$, and succeeds with high probability.
\end{theorem}

It was shown in~\cite{BFSS13} that polynomial garden-hose complexity is equivalent to log-space computation -- up to a local preprocessing of the inputs.
Instead of directly presenting garden-hose protocols, for the current construction it will be easier to argue about space-bounded algorithms and then using this equivalence as a black-box translation.
\begin{theorem}[Theorem~2.12 of \cite{BFSS13}]\label{thm:logspace}
	If $f : \{0,1\}^n \times \{0,1\}^n \to \{0,1\}$ is log-space computable, then $\gh(f)$ is polynomial in $n$.
\end{theorem}

Our attack will involve the computation of the unitary $U = \prod_{i=1}^t u_i v_i$ in the garden-hose protocol. This is a simple function, but so far we have only defined the garden-hose model for functions with a binary output.
Therefore we define an extension of the garden-hose model to functions with a larger output range, where instead of letting the water exit at Alice's or Bob's side,
we aim to let the water exit at correctly \emph{labeled pipe}. A short proof of the following proposition is given after the proof of the main theorem.
\begin{proposition}\label{prop:multi-output}
	Let $f: \{0,1\}^n \times \{0,1\}^n \to \{0,1\}^k$ be a function, such that $f$ is log-space computable and $k$ is at most $O(\log k)$. Then there exists a garden-hose protocol
	which uses a polynomial number of pipes, and such that for any input $x,y$ the water exists at Alice's side, at a pipe labeled by the output of $f(x,y)$.
\end{proposition}

We will also need a decomposition of arbitrary unitary operations into the Clifford+T gate set.
The Solovay--Kitaev theorem is a classic result which shows that any single-qubit quantum gate
can be approximated up to precision $\eps$ using $O(\log^c(1/\eps))$ gates from a finite
gate set, where $c$ is approximately equal to 2. See for example~\cite{NC00} for an exposition of the proof. Our constructions use a very particular gate set and we are only
concerned with the number of $\T$ gates instead of the total number of gates.
A recent result by Selinger strengthens the Solovay--Kitaev theorem for this specific case~\cite{Sel15}\footnote{When the single-qubit
	unitary is a z-rotation, an even stronger version of the theorem is available~\cite{RS14}.}.

\begin{theorem}[Selinger 2015]\label{thm:selingertcount}
	Any single-qubit unitary can be approximated, up
	to any given error threshold $\epsilon > 0$, by a product
	of Clifford+$\T$ operators with $\T$-count $11+12\log(1/\epsilon)$.
\end{theorem}

With these auxiliary results in place, we can present our attack on the Interleaved Product protocol.
\begin{proof}[Proof of Theorem~\ref{thm:attackipp}]
	We will describe the actions taken for any single qubit $U \ket{b}$, with $b \in \{0,1\}$, such that the probability of error is at most $\eps$.
	The protocol will be attacked by performing these actions on each qubit, $n$ times in parallel.
	Our construction can be divided in the following four steps.
	For operators $A,B$, let $\norm{ A }$ denote the operator norm, and we use $\norm{A - B}$ as an associated distance measure.
	
	\begin{enumerate}
		\item Construct a (polynomial-sized) garden-hose protocol, with a number of pipes~$s$, where the qubit is routed to a pipe labeled with a unitary $\tilde{U}$ which is $\eps_1$-close to the total product $U$.
		\item Decompose the unitaries of all labels in terms of the Clifford+T gate set, using Theorem~\ref{thm:selingertcount}. In particular, we have a Clifford+T circuit $C$ with T-count $k = O(\log{\eps_2})$ such
		that $C$ is $\eps_2$-close to $\tilde{U}$, and therefore $C$ is at most $\eps$-close to $U$, where $\eps = \eps_1 + \eps_2$.
		\item After executing the garden-hose protocol as a series of teleportations, the state at pipe $\tilde{U}$ can be
		approximated by $X^{f_x} Z^{f_z} C \ket{\psi}$, with $f_x$ and $f_z$ functions of the connections Alice and Bob made in step 1 and their measurement outcomes.
		By the construction of Figure~\ref{fig:gh-x}, described in the proof of Lemma~\ref{lem:nonlocalphase}, the garden-hose complexities $\gh(f_x)$ and $\gh(f_z)$ are at most linear in $s$.
		
		We can now alternate between applying a single gate of the circuit
		$C^\dagger$ and using Lemma~\ref{lem:nonlocalphase}, $k$ times in total, to obtain a state which only has Pauli corrections left.
		\item After Alice measures this final state, she can broadcast the outcome to Bob. Alice and Bob also broadcast their inputs and measurement outcomes, which together determine whether to flip the outcome of Alice's final measurement.
	\end{enumerate}
	
	As the first step,
	we present a log-space computation solving the following problem (equivalent to the input of the protocol, with simplified notation): The input is given by $t$ two-by-two unitary matrices, $u_1, \dots, u_t$, and
	we output a matrix $\tilde{U}$ such that $\norm{\tilde{U} - u_t \dots u_2 u_1} \leq \eps_1$, where $\tilde{U}$ is encoded using $O(\log{t} + \log 1/\eps_1)$ bits. We can then use a simple extension of Theorem~\ref{thm:logspace} to transform this
	computation to a garden-hose protocol.
	
	Store the current intermediate outcome of the product in the memory of our computation,
	using $2\ell + 2$ bits for each entry of the two-by-two matrix, $\ell+1$ for the real and imaginary part each.
	Let $M_r$ denote the memory of our log-space computation after $r$ steps, obtained by computing the product $u_{r} M_{r-1}$ with rounding.
	Since the rounded matrix entry has a difference of at most $2^{-\ell}$ with the unrounded entry, we can write the precision loss at each step as $M_r = u_{r} M_{r-1} + \Delta_r$, where $\Delta_r$ is some matrix
	with all entries absolute value at most $2^{-\ell}$. Note that $\norm{\Delta_r} \leq 2^{-\ell+1}$. 
	
	The total error incurred by the repeated rounding can now be upper bounded by
	\begin{align*}
	\norm{M_t - u_t \dots u_2 u_1} & \leq \norm{u_t M_{t-1} + \Delta_t - u_t \dots u_2 u_1} \\
	&\leq \norm{\Delta_t} + \norm{u_t (M_{t-1} - u_{t-1} \dots u_2 u_1 ) } \\
	&\leq 2^{-\ell+1} + \norm{M_{t-1} - u_{t-1} \dots u_2 u_1 } \\
	&\leq t 2^{-\ell+1}
	\end{align*}
	Here we use that $\norm{A B} \leq \norm{A} \norm{B}$ together with the unitarity of all $u_i$. The final step is by iteratively applying the earlier steps $t$ times.
	If we choose $\ell = \log{t} + \log{1 / \eps_1} + 1$ and note that the final output $\tilde{U}$ is given by $M_t$, we obtain the bound.
	
	By application of Proposition~\ref{prop:multi-output} we can convert this log-space computation to a garden-hose protocol, using $s$~pipes, where $s$ is polynomial in $\eps_1$ and $t$.
	We then teleport the qubit back-and-forth using Bell measurements given by this garden-hose protocol.
	
	As second step, we approximate the unitaries that label each output pipe of the garden-hose protocol of the previous step. In particular, consider the pipe labeled $\tilde{U}$, and say we approximate
	$\tilde{U}$ using a Clifford+T circuit $C$. By Theorem~\ref{thm:selingertcount}, we can write $C$ using $k = 11+12\log(1/\eps_2)$ T gates, such that $\norm{\tilde{U} - C} \leq \eps_2$.
	Therefore, defining $\eps = \eps_1 + \eps_2$, we  have $\norm{U - C} \leq \eps$.
	
	We will perform the next steps for all unmeasured qubits (corresponding to
	open pipes in the garden-hose model) in parallel. After the
	simultaneous round of communication, Alice and Bob are then able
	to pick the correct qubit and ignore the others.
	
	Consider the state of the qubit after the teleportations
	chosen by the garden-hose protocol. For some functions $f_x,f_z$,
	with inputs Alice's and Bob's measurement outcomes,
	the qubit has state $\X^{f_x} \Z^{f_z} U \ket{b}$. From now on, we will assume this state is exactly equal to $\X^{f_x} \Z^{f_z} C \ket{b}$ -- since $U$ is $\eps$-close
	to $C$ in the operator norm, this assumption adds error probability at most $2\eps$ to the final measurement outcome\footnote{See for instance \cite[Box 4.1]{NC00} for a computation of this added error.}.
	
	Write the inverse of this circuit as alternation between gates from the Clifford group and $\T$ gates, 
	$C^\dagger=C_{k} \T C_{k-1} \T \dots C_{1} \T C_{0} $. 
	We will remove $C$ from the qubit by 
	applying these gates, one by one, by repeated application of Lemma~\ref{lem:nonlocalphase}.
	As convenient shorthand, define the state of the qubit after applying the first $r$ layers of $C^{\dagger}$, i.e.\ up to and including $C_r$, of $C^{\dagger}$ as 
	\[
	\ket{\psi_r} = \T^{\dagger} C^{\dagger}_{r+1} \T^{\dagger} C_{r+2} \dots \T^{\dagger} C_{k}^\dagger \ket{b} \,.
	\]
	In particular, we have $C_r \T \ket{\psi_{r-1}} = \ket{\psi_r}$.
	
	By exactly the same construction used in the proof of Lemma~\ref{lem:nonlocalphase}, shown in Figure~\ref{fig:gh-x}, we observe
	that the garden-hose complexities of the functions $f_x$ and $f_z$ is at most $2s+1$.
	That is, the protocol uses 2 pipes for all of the $s$ EPR pairs, and connects them in parallel if the corresponding X- or Z-correction is 0, or crosswise if the corresponding X- or Z-correction is 1.
	
	We will use divide $f^r_x$ and  $f^r_z$ as the functions describing the X and Z corrections at the end of the step $r$. Define $m_r = \max \{\gh(f^i_x),\gh(f^i_z) \}$ to
	be the maximum garden-hose complexity out the of functions describing the $\X$ and $\Z$ corrections after step $r$.
	After Alice executes the Clifford gate $C_0$, the new key functions $f^0_x$ and $f^0_z$ can be written as (the NOT of) an XOR of subsets of the previous keys, e.g., one of the keys could 
	be $f_x \oplus f_z$. By Lemma~\ref{lem:sumgh}, we then have that our starting complexities $\gh(f^0_x)$ and $\gh(f^0_z)$ are at most linear in $s$.
	
	Now, for any layer $r=1,2,\dots,k$: Our qubit starts in the state $\X^{f^{r-1}_x} \Z^{f^{r-1}_z} \ket{\psi_{r-1}}$,
	for some functions $f^{r-1}_x, f^{r-1}_z$ that each have garden-hose complexity at most $m_{r-1}$. After
	Alice performs a $\T$ gate, the qubit is in the state 
	\[
	\T \X^{f^{r-1}_x} \Z^{f^{r-1}_z} \ket{\psi_{r-1}}
	=
	\P^{f^{r-1}_x} \X^{f^{r-1}_x} \Z^{f^{r-1}_z} \T \ket{\psi_{r-1}} \,.
	\]
	Now, we apply Lemma~\ref{lem:nonlocalphase}, costing $2 \gh(f^{r-1}_x)$ EPR pairs, so that Alice has the state 
	\[
	\X^{f^{r-1}_x \oplus g_r} \Z^{f^{r-1}_z \oplus h_r} \T \ket{\psi_{r-1}} \,,
	\]
	for some functions $g_r$ and $h_r$ that depend on the measurement results by Alice and Bob. We have that $\gh(g_r) \leq 4 \gh(f^{r-1}_x) +1$ and $\gh(g_r) \leq 11 \gh(f^{r-1}_x) +2$.
	
	Now Alice applies the Clifford group gate $C_{r}$, so that the state becomes
	\[
	C_{r} \X^{f^{r-1}_x \oplus g_r} \Z^{f^{r-1}_z \oplus h_r} \T \ket{\psi_{r-1}}
	=
	\X^{f^{r}_x} \Z^{f^{r}_z} \ket{\psi_{r}}
	\,.
	\]
	The functions $f^{r}_x$ and $f^{r}_z$ can be expressed as XOR of the functions $f^{r-1}_x$, $f^{r-1}_y$, $g_r$, $h_r$. These functions have garden-hose complexity respectively at most $m_{r-1}$, $m_{r-1}$, $4 m_{r-1} + 1$ and $11 m_{r-1} + 2$.
	By application of Lemma~\ref{lem:sumgh}, the exclusive OR of these functions therefore at most has garden-hose complexity $m_r \leq 4(m_{r-1} + m_{r-1} + 4 m_{r-1} + 1 + 11 m_{r-1} + 2) + 1 = 68 m_{r-1} + 13$.
	
	Finally, after application of the gates in $\C^\dagger$, Alice has a qubit in a state which is $\eps$-close to $\X^{f^r_x} \Z^{f^r_z} \ket{b}$. Measurement in the computational basis will produce outcome $b \oplus f^r_x$ with high probability.
	Besides this final measurement, Alice and Bob both broadcast all teleportation measurement outcomes in their step of simultaneous communication. From these outcomes they can each locally compute $f^r_x$ and so derive the bit $b$ from the outcome, which equals $b \oplus f^r_x$, 
	breaking the protocol.
	
	Our total entanglement usage is $s$ for the first step, and then for each of the at most $s$ output pipes, Alice performs the rest of the protocol. For the part of the protocol that undoes the unitary $U$, we use at
	most $2 \sum_{r=0}^{k-1} m_r$ EPR pairs (for each of the at most $s$ output pipes of the first part). We have $m_0 \leq O(s)$ and $m_r \leq m_0 \cdot 2^{O(k)}$. Since $s$ is polynomial in $t$ and $\eps_1$ and $k=O(\log \eps_2)$, the total protocol uses
	entanglement polynomial in $t$ and $\eps$.
\end{proof}

Our attack replaces the exponential dependence on $t$ of the attacks presented in~\cite{CL15} by a polynomial dependence. For the case of $\etaerr = 0$, we would
need an error per qubit of around $\frac{\eps}{n}$ to achieve total error at most $\eps$. In that case, the entanglement required still grows as a polynomial, now with a super-linear dependence of both parameters $n$ and $t$.

Only the first step of our attack, i.e.\ the garden-hose protocol which computes a unitary from the inputs of the players, is specific to the interleaved product protocol. 
This attack can therefore be seen as a blueprint for attacks on a larger class of protocols: any protocol of this same form, where the unitary operation chosen depends
on a log-space computable function with classical inputs, can be attacked with entanglement which scales as a polynomial in the size of the classical inputs.

\begin{proof}[Proof of Proposition~\ref{prop:multi-output}]
	We can split up the computation $f: \{0,1\}^n \times \{0,1\}^n \to \{0,1\}^k$ into $k$ functions that each compute a bit,
	$f_1, \dots, f_k$. Since $f$ is a log-space computation, each of these functions is also a log-space computation and therefore has a polynomial-size garden-hose protocol by Theorem~\ref{thm:logspace}.
	Using Lemma~\ref{lem:ghsingleoutput}, we can with linear overhead transform each of these protocol into a unique-output protocol,
	so that the water flows out at a unique pipe when the function is 0 and another unique pipe when the function is 1.
	Let $p$ be a polynomial so that the single-output garden-hose protocol of each function $f_i$ uses pipes at most $p(n)$.

	First use the protocol for $f_1$, with output pipes labeled $0$ and $1$. Now each of these output pipes we feed into their own copy of $f_2$.
	The 0 output of the first copy we label $00$ and its 1 output $10$. Similarly, we label the 0 output of the second copy $01$ and the 1 output we label $11$.
	By recursively continuing this construction, we build a garden-hose protocol for the function $f$ which uses $s$ pipes, where $s$ is at most
	\[
	s \leq \sum_{i=1}^k 2^{i-1} p(n) \leq 2^k p(n) \,.
	\]
	Since we have taken $k=O(\log n)$, this construction uses a number of pipes polynomial in $n$.
\end{proof}

\section{Discussion}\label{sec:tdepthdisc}
We combined ideas from the garden-hose model with techniques from quantum cryptography
to find a class of quantum circuits for which instantaneous non-local computation is efficient.
These constructions can be used as attacks on protocols for quantum position-verification, and 
could also be translated back into the settings related to physics (most notable the relation between the constraints of relativity theory and quantum measurements) and distributed computing.

The resource usage of
instantaneous non-local quantum computation quantifies the non-locality 
present in a bi- or multi-partite quantum operation, and there is still room for new upper and lower bounds. Any such bounds will result in new insights, both
in terms of position-based quantum cryptography, but also in the other mentioned settings.

Some possible approaches for continuing this line of research are as follows:
\begin{itemize}
\item Computing the Pauli corrections happens without error in our current construction. Perhaps introducing randomness and a small
probability of error -- or the usage of entanglement as given in the \emph{quantum garden-hose model}
of \cite[Section 2.5]{BFSS13} -- could make this scheme more efficient.
\item Future research might be able to extend this type of construction to a wider gate set or model of computation. One could think for example of a Clifford+cyclotomic gate set~\cite{FGKM15}, match-gate computation~\cite{JKMW09},
or measurement-based quantum computation~\cite{BBD+09,BFK09}.
\item We presented an attack on the Interleaved Product protocol which required entanglement polynomial in $t$. Since the exponent of this polynomial was quite large, the scheme could still
be secure under realistic assumptions. Since the parameter $t$ concerns the \emph{classical} information that the verifiers send, requiring attackers to manipulate an amount of entanglement
which scales linearly with the classical information would already make a scheme unpractical to break in practice -- let alone a quadratic or cubic dependence. 
\item The combination of the garden-hose model with the tool set of
blind quantum computation is potentially powerful in other settings.
For example, following up on Broadbent and Jeffery who
published constructions
for quantum homomorphic encryption for circuits of low T-gate complexity~\cite{BJ15}, Dulek, Speelman, and Schaffner~\cite{DSS16}
developed a scheme for quantum homomorphic encryption,
based on this combination as presented in (a preprint of) this work.
\end{itemize}

\section*{Acknowledgments}
The author is supported by the EU projects SIQS and QALGO, and thanks Anne Broadbent, Harry Buhrman, Yfke Dulek and Christian Schaffner for useful discussions.

\appendix

\newcommand{\state}{\mathcal{S}}
\newcommand{\complex}{\mathbb{C}}
\section{Definition of INQC}\label{sec:definqc}
An \emph{instantaneous non-local quantum protocol that uses $k$ qubits of entanglement}\index{instantaneous non-local quantum computation|textbf} is a protocol of the following form. 

Alice and Bob start with a fixed, chosen $2k$-qubit state $\eta_{A_{e} B_{e}} \in \complex^{2^k} \otimes \complex^{2^k}$, the entanglement. (Our protocols all use the special case where this state is a tensor product of $k$ EPR pairs.)
\index{EPR pair}\index{entanglement}
The players receive an input state $\rho \in \state(A_{in} \otimes B_{in})$,
where $\state(A)$\index{S@$\state$} is used for the set of density matrices on some Hilbert space $A$. Let $A_{m}, A_{s}, B_{m}, A_{s}$ denote arbitrary-sized quantum registers.
Alice applies some quantum operation, i.e.\ completely positive trace-preserving map, $\mathcal{A_1} : \state(A_{in} \otimes A_{e}) \to \state(A_{m} \otimes A_{s})$ and Bob
applies the quantum operation $\mathcal{B_1} : \state(B_{in} \otimes B_{e}) \to \state(B_{m} \otimes B_{s})$.
Alice sends the register $A_{s}$ to Bob, while simultaneously Bob sends $B_{s}$ to Alice.

Afterwards Alice applies the quantum operation $\mathcal{A_2} : \state(A_{m} \otimes B_{s}) \to \state(A_{out})$ on her memory and the state she received from Bob, and outputs the result.
Likewise Bob applies the operation $\mathcal{B_2} : \state(B_{m} \otimes A_{s}) \to \state(B_{out})$ on the part of the quantum state he kept and outputs the result of this operation.

\begin{definition}

Let $\Phi : \state(A_{in} \otimes B_{in}) \to \state(A_{out} \otimes B_{out})$ be a bipartite quantum operation, i.e.\ a completely positive trace-preserving map,
for some input registers $A_{in}, B_{in}$ and output registers $A_{out}, B_{out}$.

We say that $\INQC_{\eps}(\Phi)$ is the smallest number $k$ such that there exists
an instantaneous non-local quantum protocol that uses $k$ qubits of entanglement, with induced channel $\Psi : \state(A_{in} \otimes B_{in}) \to \state(A_{out} \otimes B_{out})$, so that $\norm{\Phi - \Psi}_{\diamond} \leq \eps$.
\end{definition}

For any unitary $U$, we write $\INQC_{\eps}(U)$ as a shorthand for  $\INQC_{\eps}(\Phi_{U})$, where $\Phi_{U}$ is the induced quantum operation
defined by $\rho_{AB} \to U \rho_{AB} U^\dagger$. In this paper, we assume for simplicity that Alice's and Bob's input and output registers all consist of $n$ qubits.

These definitions are mostly compatible with those given in \cite{Beigi2011}, 
but differ in two ways -- both are unimportant for our results in this paper, but
might be relevant for follow-up results, especially when proving lower bounds.
Firstly, we made the choice for generality to
allow the players to communicate using qubits, instead of just classical messages.
As long as the number of communicated qubits
is not too large, quantum communication could potentially be replaced by classical
communication using teleportation, at the cost of extra entanglement -- the counted resource.
Secondly, we make the choice to explicitly separate the shared entangled state from the local memory in notation -- Beigi and K\"onig split the state in a measured and unmeasured part, but do not introduce notation
for (free) extra local memory in addition to the shared entangled state.

Whether these choices are reasonable or not will also depend on the exact application. Since we mostly think about applications to position-based quantum cryptography,
giving the players, i.e.\ `attackers', as much power as possible seems the most natural.

\section{The Clifford hierarchy}\index{Clifford hierarchy}\label{sec:proofcliffordhierachy}
The Clifford hierarchy, also called the Gottesman--Chuang hierarchy,
generalizes the definition of the Clifford group of Equation~\ref{eq:defclifford} in the following way~\cite{GC99}.
Define $\C_1 = \mathcal{P}$, the first level of the hierarchy, as the Pauli group.
Recursively define the $k$-th level as
\[\C_{k} = \{U \in U(2^n) \mid \forall \sigma \in \mathcal{P} : U \sigma U^\dagger \in \C_{k-1}\} \,.
\]
Then $\C_{2}$ is the Clifford group and the next levels consist of increasingly more quantum operations -- although for $k \geq 3$ the
set $\C_k$ is no longer a group~\cite{ZCC2008}.

The method behind the protocol of Theorem~\ref{thm:tcount} immediately translates to the related setting of the Clifford hierarchy.
Since the dependence on $n$ is exponential, 
 Proposition~\ref{prop:cliffordhierarchy} will only be a
qualitative improvement over Beigi and K\"onig's port-based teleportation construction when both $n$ and the level $k$ are small.

The results of Chakraborty and Leverrier~\cite{CL15} contain a complete proof of  Proposition~\ref{prop:cliffordhierarchy}, proven
independently and made available earlier than (the preprint of) the current paper.
We still include a proof of the statement
as an illustrative application of the proof technique of Section~\ref{sec:tcount}.

\begin{proposition}\label{prop:cliffordhierarchy}
Let $U$ be an $n$-qubit operation in the $k$-th level of the Clifford hierarchy, where Alice receives $n/2$ qubits and Bob receives $n/2$ 
qubits, then
$\INQC(U) \leq O(n 4^{n k})$.
\end{proposition}
\begin{proof}[Proof sketch]
First Bob teleports his qubits to Alice,
with $n$ outcomes for $\X$ and $\Z$. Alice applies $U$ to
the uncorrected state, so that now the state equals
$U \X^{b_x} \Z^{b_z} \ket{\psi} = V_{b_x,b_z} U \ket{\psi}$,
where $V_{b_x,b_z}$ is an operator in the $(k-1)$-th level of the
Clifford hierarchy. Exactly which operator depends on Bob's measurement outcomes $b_x, b_z$.

Alice teleports the entire state to Bob, with outcomes $a_x, a_z$, 
and Bob applies the inverse $V^{\dagger}_{b_x,b_z}$, so that the state
is 
\[
V^{\dagger}_{b_x,b_z} \X^{a_x} \Z^{a_z} V_{b_x,b_z} U \ket{\psi}
= W_{a_x,a_z, b_x,b_z} U \ket{\psi} \, ,
\]
with $W_{a_x,a_z, b_x,b_z}$ in the $(k-2)$-th level of the Clifford hierarchy.
For every possible value of $b_x,b_z$, the players
share a set of $n$ EPR pairs. Bob teleports the state using
the set labeled with his measurement outcome $b_x,b_z$, obtaining teleportation corrections $\hat{b}_{x} , \hat{b}_{z}$.

For every set the players repeat this protocol recursively, in the following way.
For any set, Alice repeats the protocol as if it were the set used by Bob.
At the correct set, Alice effectively knows the values ${b_x, b_z}$ from the label, 
and $a_x, a_z$ she knows as own measurement outcomes. The state present is $\X^{\hat{b}_{x}} \Z^{\hat{b}_{z}} W_{a_x,a_z, b_x,b_z} U \ket{\psi}$.
When Alice applies $W^\dagger_{a_x,a_z, b_x,b_z}$, the state is given by $F_{a_x,a_z, b_x,b_z, \hat{b}_{x} , \hat{b}_{z}}  U \ket{\psi}$,
with $F$ in the $(k-3)$-th level of the Clifford hierarchy. Of this state, effectively only $\hat{b}_{x} , \hat{b}_{z}$ is unknown to Alice.
Alice teleports this state to Bob using the EPR pairs
labeled with $a_x,a_z$, and the recursive step is complete.

The players continue these steps until the first level of the hierarchy is reached -- formed by Pauli operators -- after which they can exchange the outcomes of their measurements to undo these and obtain $U \ket{\psi}$.

After $t$ steps, 
Every teleportation step after the first uses a set of $n$ EPR pairs, picked out of $4^{n}$ possibilities corresponding to
the Pauli correction of the $n$ qubits teleported in the previous step.

Summing over all rounds gives a total entanglement use of
$n \sum_{t=1}^k 4^{n t} = O(n 4^{n k})$.
\end{proof}

\section{Proof of Lemma~\ref{lem:sumgh}: Garden-hose protocols for XOR of functions}\label{sec:proofsumgh}
\textbf{To prove:}
Let $(f_1, f_2, \dots, f_k)$ be functions, where each function $f_i$ has garden-hose complexity $\gh(f_i)$.
Let $c \in \{0,1\}$ be an arbitrary bit that is 0 or 1. 
Then, 
\[
\gh \left( c \oplus \bigoplus _{i=1} ^{k} f_i \right) \leq 4 \sum_{i=1}^k \gh(f_i) + 1 \, .
\]

\begin{proof}[Proof sketch]
This statement was proven by Klauck and Podder~\cite[Theorem 18]{KP14} in a more general form, using the following two steps: First, any garden-hose protocol can be turned into a single-output garden-hose protocol,
repeated in this paper as Lemma~\ref{lem:ghsingleoutput}, such that the new complexity is at most three times the old complexity.
Then, these single-output garden-hose protocols can be used as nodes in a permutation branching program. Our current case is simply an instantiation of that proof for the particular case of the exclusive OR, 
together with the observation that we can combine both steps into one for this particular case.

For all functions $f_i$ we build a gadget with two input pipes and two output pipes, such that if the water flows
in at input pipe labeled $b \in \{0,1\}$, it flows out at the pipe labeled $f_i \oplus b$.
See Figure~\ref{fig:xorgadget} for an overview. We use four copies of the garden-hose protocol for $f_i$.

The open 0 output pipes of the protocol for $f_i$ in copy 0-IN$_i$ are connected to the open 0 output pipes in copy 0-OUT$_i$. The designated source pipe of the original protocol for $f_i$ in copy
0-OUT$_i$ is then guaranteed to be the output.\footnote{This same trick is used in the proof of Lemma~\ref{lem:ghsingleoutput} in~\cite[Lemma 11]{KP14} and in our proof of Lemma~\ref{lem:nonlocalphase}.}
We similarly connect the 1 outputs of 0-IN$_i$ to the 1 outputs of 1-OUT$_i$. This construction, i.e.\ before adding the 1-IN copy, is exactly the method used to create a single-output protocol.
We connect the open 0 pipes of 1-IN$_i$ to the open 0 pipes of 1-OUT$_i$ and the open 1 pipes of the open 1 pipes of 1-IN$_i$ to the open 1 pipes of 0-OUT$_i$.

The gadget then works as claimed by direct inspection. Since all four copies are wired exactly the same, the path of the water through the `OUT' copy is the reverse of the path it followed through
the `IN' copy, and therefore the water will exit correctly -- at the pipe which was the source of the original protocol.
\begin{figure}
\begin{center}
\begin{tikzpicture}
\coordinate (in0) at (0,8);
\coordinate (in1) at (4,8);

\coordinate (tl1) at (1,7);
\coordinate (br1) at (3,5);

\coordinate (tl2) at (5,7);
\coordinate (br2) at (7,5);

\coordinate (tl3) at (1,3);
\coordinate (br3) at (3,1);

\coordinate (tl4) at (5,3);
\coordinate (br4) at (7,1);

\coordinate (out0) at (0,0);
\coordinate (out1) at (4,0);

\draw[thick] (tl1) rectangle node[align=center] {protocol\\for $f_i$\\ \\0-IN$_i$} (br1);
\draw[thick] (tl2) rectangle node[align=center] {protocol\\for $f_i$\\ \\$\text{1-IN}_i$} (br2);
\draw[thick] (tl3) rectangle node[align=center] {protocol\\for $f_i$\\ \\$\text{0-OUT}_i$} (br3);
\draw[thick] (tl4) rectangle node[align=center] {protocol\\for $f_i$\\ \\$\text{1-OUT}_i$} (br4);

\draw ($ (tl1) + (0, -1.5) $) to[out=200, in=160] ($ (tl3) + (0, -0.5) $);
\draw ($ (tl1) + (0, -1.6) $) to[out=200, in=160] ($ (tl3) + (0, -0.6) $);
\draw ($ (tl1) + (0, -1.7) $) to[out=200, in=160] ($ (tl3) + (0, -0.7) $);

\draw ($ (tl2) + (0, -1.5) $) to[out=200, in=160] ($ (tl4) + (0, -0.5) $);
\draw ($ (tl2) + (0, -1.6) $) to[out=200, in=160] ($ (tl4) + (0, -0.6) $);
\draw ($ (tl2) + (0, -1.7) $) to[out=200, in=160] ($ (tl4) + (0, -0.7) $);

\draw ($ (br1) + (0, 0.5) $) to[out=-20, in=20] ($ (br4) + (0, 1.5) $);
\draw ($ (br1) + (0, 0.4) $) to[out=-20, in=20] ($ (br4) + (0, 1.4) $);

\draw ($ (br2) + (0, 0.5) $) to[out=-20, in=20] ($ (br3) + (0, 1.5) $);
\draw ($ (br2) + (0, 0.4) $) to[out=-20, in=20] ($ (br3) + (0, 1.4) $);

\draw[dashed] ($ (in0) $) node[anchor=south] {0 in} to[out=-90, in=160] ($ (tl1) + (0, -0.5) $);
\draw[dashed] ($ (in1) $) node[anchor=south] {1 in} to[out=-90, in=160] ($ (tl2) + (0, -0.5) $);
\draw[dashed] ($ (tl3) + (0, -1.5) $)  to[out=200, in=90] ($ (out0) $) node[anchor=north] {0 out}  ;
\draw[dashed] ($ (tl4) + (0, -1.5) $) to[out=200, in=90] ($ (out1) $) node[anchor=north] {1 out}  ;

\end{tikzpicture}
\end{center}
\caption{XOR gadget for any function $f_i$, total complexity $4 \gh(f_i)$. }
\label{fig:xorgadget}
\end{figure}
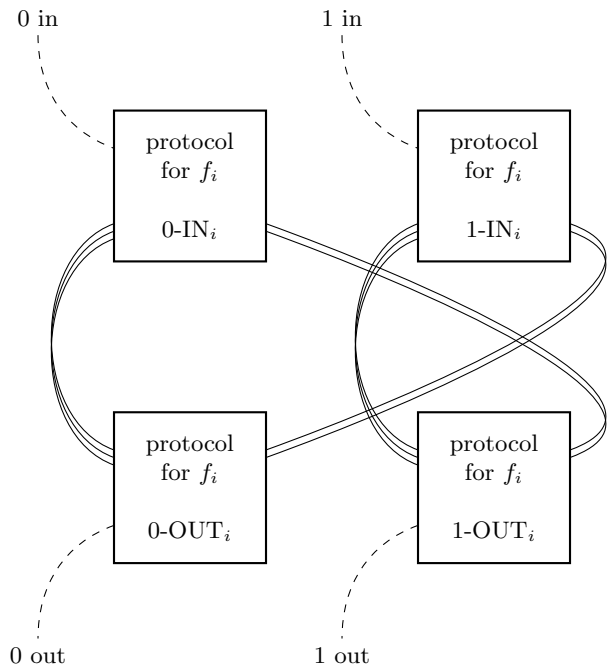
\end{proof}

\bibliographystyle{plainurl}
\bibliography{nonlocal-arxiv}

\begin{thebibliography}{10}

\bibitem{AG04}
Scott Aaronson and Daniel Gottesman.
\newblock Improved simulation of stabilizer circuits.
\newblock {\em Physical Review A}, 70(5):052328, 2004.

\bibitem{AMM14}
Matthew Amy, Dmitri Maslov, and Michele Mosca.
\newblock Polynomial-time {T}-depth optimization of {Clifford+T} circuits via
  matroid partitioning.
\newblock {\em Computer-Aided Design of Integrated Circuits and Systems, IEEE
  Transactions on}, 33(10):1476--1489, Oct 2014.
\newblock \href {http://dx.doi.org/10.1109/TCAD.2014.2341953}
  {\path{doi:10.1109/TCAD.2014.2341953}}.

\bibitem{AMMR13}
Matthew Amy, Dmitri Maslov, Michele Mosca, and Martin Roetteler.
\newblock A meet-in-the-middle algorithm for fast synthesis of depth-optimal
  quantum circuits.
\newblock {\em Trans. Comp.-Aided Des. Integ. Cir. Sys.}, 32(6):818--830, June
  2013.
\newblock \href {http://dx.doi.org/10.1109/TCAD.2013.2244643}
  {\path{doi:10.1109/TCAD.2013.2244643}}.

\bibitem{AS06}
Pablo Arrighi and Louis Salvail.
\newblock Blind quantum computation.
\newblock {\em International Journal of Quantum Information}, 4(05):883--898,
  2006.

\bibitem{Beigi2011}
Salman Beigi and Robert K\"onig.
\newblock Simplified instantaneous non-local quantum computation with
  applications to position-based cryptography.
\newblock {\em New Journal of Physics}, 13(9):093036, 2011.

\bibitem{BBD+09}
HJ~Briegel, DE~Browne, W~D{\"u}r, R~Raussendorf, and M~Van~den Nest.
\newblock Measurement-based quantum computation.
\newblock {\em Nature Physics}, 5(1):19--26, 2009.

\bibitem{Bro15}
Anne Broadbent.
\newblock Delegating private quantum computations.
\newblock {\em Canadian Journal of Physics}, 93(9):941--946, 2015.
\newblock \href {http://dx.doi.org/10.1139/cjp-2015-0030}
  {\path{doi:10.1139/cjp-2015-0030}}.

\bibitem{Bro15a}
Anne Broadbent.
\newblock {P}opescu--{R}ohrlich correlations imply efficient instantaneous
  nonlocal quantum computation.
\newblock {\em arXiv preprint arXiv:1512.04930}, 2015.

\bibitem{BFK09}
Anne Broadbent, Joseph Fitzsimons, and Elham Kashefi.
\newblock Universal blind quantum computation.
\newblock In {\em Foundations of Computer Science, 2009. FOCS'09. 50th Annual
  IEEE Symposium on}, pages 517--526. IEEE, 2009.

\bibitem{BJ15}
Anne Broadbent and Stacey Jeffery.
\newblock Quantum homomorphic encryption for circuits of low {T-gate}
  complexity.
\newblock In Rosario Gennaro and Matthew Robshaw, editors, {\em Advances in
  Cryptology -- CRYPTO 2015}, volume 9216 of {\em Lecture Notes in Computer
  Science}, pages 609--629. Springer Berlin Heidelberg, 2015.
\newblock \href {http://dx.doi.org/10.1007/978-3-662-48000-7_30}
  {\path{doi:10.1007/978-3-662-48000-7_30}}.

\bibitem{BCL+06}
H.~Buhrman, R.~Cleve, M.~Laurent, N.~Linden, A.~Schrijver, and F.~Unger.
\newblock New limits on fault-tolerant quantum computation.
\newblock In {\em Foundations of Computer Science, 2006. FOCS '06. 47th Annual
  IEEE Symposium on}, pages 411--419, Oct 2006.
\newblock \href {http://dx.doi.org/10.1109/FOCS.2006.50}
  {\path{doi:10.1109/FOCS.2006.50}}.

\bibitem{Buhrman2011}
Harry Buhrman, Nishanth Chandran, Serge Fehr, Ran Gelles, Vipul Goyal, Rafail
  Ostrovsky, and Christian Schaffner.
\newblock Position-based quantum cryptography: Impossibility and constructions.
\newblock In Phillip Rogaway, editor, {\em Advances in Cryptology - CRYPTO
  2011}, volume 6841 of {\em Lecture Notes in Computer Science}, pages
  429--446. Springer Berlin / Heidelberg, 2011.

\bibitem{BFSS13}
Harry Buhrman, Serge Fehr, Christian Schaffner, and Florian Speelman.
\newblock The garden-hose model.
\newblock In {\em Proceedings of the 4th Conference on Innovations in
  Theoretical Computer Science}, ITCS '13, pages 145--158, New York, NY, USA,
  2013. ACM.
\newblock \href {http://dx.doi.org/10.1145/2422436.2422455}
  {\path{doi:10.1145/2422436.2422455}}.

\bibitem{CL15}
Kaushik Chakraborty and Anthony Leverrier.
\newblock Practical position-based quantum cryptography.
\newblock {\em Phys. Rev. A}, 92:052304, Nov 2015.
\newblock \href {http://dx.doi.org/10.1103/PhysRevA.92.052304}
  {\path{doi:10.1103/PhysRevA.92.052304}}.

\bibitem{Childs2005}
Andrew~M Childs.
\newblock Secure assisted quantum computation.
\newblock {\em Quantum Information \& Computation}, 5(6):456--466, 2005.

\bibitem{CCJP10}
S~R Clark, A~J Connor, D~Jaksch, and S~Popescu.
\newblock Entanglement consumption of instantaneous nonlocal quantum
  measurements.
\newblock {\em New Journal of Physics}, 12(8):083034, 2010.
\newblock URL: \url{http://stacks.iop.org/1367-2630/12/i=8/a=083034}.

\bibitem{DSS16}
Yfke Dulek, Christian Schaffner, and Florian Speelman.
\newblock Quantum homomorphic encryption for polynomial-sized circuits.
\newblock {\em arXiv preprint arXiv:1603.09717}, 2016.

\bibitem{DNS10}
Fr\'{e}d\'{e}ric Dupuis, Jesper~Buus Nielsen, and Louis Salvail.
\newblock {Secure two-party quantum evaluation of unitaries against specious
  adversaries}.
\newblock In {\em CRYPTO}, pages 685--706, September 2010.
\newblock \href {http://arxiv.org/abs/1009.2096} {\path{arXiv:1009.2096}},
  \href {http://dx.doi.org/10.1007/978-3-642-14623-7_37}
  {\path{doi:10.1007/978-3-642-14623-7_37}}.

\bibitem{FBS+14}
KAG Fisher, A~Broadbent, LK~Shalm, Z~Yan, J~Lavoie, R~Prevedel, T~Jennewein,
  and KJ~Resch.
\newblock Quantum computing on encrypted data.
\newblock {\em Nature communications}, 5, 2014.

\bibitem{FGKM15}
Simon Forest, David Gosset, Vadym Kliuchnikov, and David McKinnon.
\newblock Exact synthesis of single-qubit unitaries over {Clifford-cyclotomic}
  gate sets.
\newblock {\em Journal of Mathematical Physics}, 56(8):--, 2015.
\newblock \href {http://dx.doi.org/http://dx.doi.org/10.1063/1.4927100}
  {\path{doi:http://dx.doi.org/10.1063/1.4927100}}.

\bibitem{GS13}
Brett Giles and Peter Selinger.
\newblock Exact synthesis of multiqubit {Clifford+T} circuits.
\newblock {\em Physical Review A}, 87(3):032332, 2013.

\bibitem{Gottesman98}
Daniel Gottesman.
\newblock {The Heisenberg representation of quantum computers}.
\newblock In {\em {Group theoretical methods in physics. Proceedings, 22nd
  International Colloquium, Group22, ICGTMP'98, Hobart, Australia, July 13-17,
  1998}}, 1998.
\newblock \href {http://arxiv.org/abs/quant-ph/9807006}
  {\path{arXiv:quant-ph/9807006}}.

\bibitem{Got98}
Daniel Gottesman.
\newblock Theory of fault-tolerant quantum computation.
\newblock {\em Phys. Rev. A}, 57:127--137, Jan 1998.
\newblock \href {http://dx.doi.org/10.1103/PhysRevA.57.127}
  {\path{doi:10.1103/PhysRevA.57.127}}.

\bibitem{GC99}
Daniel Gottesman and Isaac~L. Chuang.
\newblock {Quantum Teleportation is a Universal Computational Primitive}.
\newblock {\em Nature}, 402:390--393, August 1999.
\newblock \href {http://arxiv.org/abs/9908010} {\path{arXiv:9908010}}, \href
  {http://dx.doi.org/10.1038/46503} {\path{doi:10.1038/46503}}.

\bibitem{IH08}
Satoshi Ishizaka and Tohya Hiroshima.
\newblock Asymptotic teleportation scheme as a universal programmable quantum
  processor.
\newblock {\em Phys. Rev. Lett.}, 101(24):240501, Dec 2008.
\newblock \href {http://dx.doi.org/10.1103/PhysRevLett.101.240501}
  {\path{doi:10.1103/PhysRevLett.101.240501}}.

\bibitem{IH09}
Satoshi Ishizaka and Tohya Hiroshima.
\newblock Quantum teleportation scheme by selecting one of multiple output
  ports.
\newblock {\em Phys. Rev. A}, 79(4):042306, Apr 2009.
\newblock \href {http://dx.doi.org/10.1103/PhysRevA.79.042306}
  {\path{doi:10.1103/PhysRevA.79.042306}}.

\bibitem{JKMW09}
Richard Jozsa, Barbara Kraus, Akimasa Miyake, and John Watrous.
\newblock Matchgate and space-bounded quantum computations are equivalent.
\newblock In {\em Proceedings of the Royal Society of London A: Mathematical,
  Physical and Engineering Sciences}, page rspa20090433. The Royal Society,
  2009.

\bibitem{KMS11}
Adrian Kent, William~J. Munro, and Timothy~P. Spiller.
\newblock Quantum tagging: Authenticating location via quantum information and
  relativistic signaling constraints.
\newblock {\em Phys. Rev. A}, 84:012326, Jul 2011.
\newblock \href {http://dx.doi.org/10.1103/PhysRevA.84.012326}
  {\path{doi:10.1103/PhysRevA.84.012326}}.

\bibitem{KP14}
Hartmut Klauck and Supartha Podder.
\newblock New bounds for the garden-hose model.
\newblock In Venkatesh Raman and S.~P. Suresh, editors, {\em 34th International
  Conference on Foundation of Software Technology and Theoretical Computer
  Science (FSTTCS 2014)}, volume~29 of {\em Leibniz International Proceedings
  in Informatics (LIPIcs)}, pages 481--492, Dagstuhl, Germany, 2014. Schloss
  Dagstuhl--Leibniz-Zentrum fuer Informatik.
\newblock \href
  {http://dx.doi.org/http://dx.doi.org/10.4230/LIPIcs.FSTTCS.2014.481}
  {\path{doi:http://dx.doi.org/10.4230/LIPIcs.FSTTCS.2014.481}}.

\bibitem{KMM13}
Vadym Kliuchnikov, Dmitri Maslov, and Michele Mosca.
\newblock Fast and efficient exact synthesis of single-qubit unitaries
  generated by {Clifford} and {T} gates.
\newblock {\em Quantum Info. Comput.}, 13(7-8):607--630, July 2013.

\bibitem{LL11}
Hoi-Kwan Lau and Hoi-Kwong Lo.
\newblock Insecurity of position-based quantum-cryptography protocols against
  entanglement attacks.
\newblock {\em Phys. Rev. A}, 83(1):012322, Jan 2011.
\newblock \href {http://dx.doi.org/10.1103/PhysRevA.83.012322}
  {\path{doi:10.1103/PhysRevA.83.012322}}.

\bibitem{NRS01}
Gabriele Nebe, Eric~M. Rains, and Neil J.~A. Sloane.
\newblock The invariants of the {Clifford} groups.
\newblock {\em Designs, Codes and Cryptography}, 24(1):99--122, 2001.
\newblock \href {http://dx.doi.org/10.1023/A:1011233615437}
  {\path{doi:10.1023/A:1011233615437}}.

\bibitem{NC00}
Michael~A. Nielsen and Isaac~L. Chuang.
\newblock {\em Quantum Computation and Quantum Information}.
\newblock Cambridge university press, 2000.

\bibitem{RS14}
Neil~J Ross and Peter Selinger.
\newblock Optimal ancilla-free {Clifford+T} approximation of z-rotations.
\newblock {\em arXiv preprint arXiv:1403.2975}, 2014.

\bibitem{Sel13}
Peter Selinger.
\newblock Quantum circuits of {T}-depth one.
\newblock {\em Physical Review A}, 87(4):042302, 2013.

\bibitem{Sel15}
Peter Selinger.
\newblock Efficient {Clifford+T} approximation of single-qubit operators.
\newblock {\em Quantum Information \& Computation}, 15(1-2):159--180, January
  2015.

\bibitem{Speelman11}
Florian Speelman.
\newblock Position-based quantum cryptography and the garden-hose game.
\newblock Master's thesis, University of Amsterdam, 2011.

\bibitem{Vaidman03}
Lev Vaidman.
\newblock Instantaneous measurement of nonlocal variables.
\newblock {\em Phys. Rev. Lett.}, 90(1):010402, Jan 2003.
\newblock \href {http://dx.doi.org/10.1103/PhysRevLett.90.010402}
  {\path{doi:10.1103/PhysRevLett.90.010402}}.

\bibitem{Yu11}
Li~Yu.
\newblock Fast controlled unitary protocols using group or quasigroup
  structures.
\newblock {\em arXiv preprint arXiv:1112.0307}, 2011.

\bibitem{YGC12}
Li~Yu, Robert~B Griffiths, and Scott~M Cohen.
\newblock Fast protocols for local implementation of bipartite nonlocal
  unitaries.
\newblock {\em Physical Review A}, 85(1):012304, 2012.

\bibitem{ZCC2008}
Bei Zeng, Xie Chen, and Isaac~L Chuang.
\newblock Semi-{C}lifford operations, structure of {$C_k$} hierarchy, and gate
  complexity for fault-tolerant quantum computation.
\newblock {\em Physical Review A}, 77(4):042313, 2008.

\end{thebibliography}

\end{document}